\documentclass[12pt]{article}
\usepackage[T1]{fontenc}
\usepackage[utf8]{inputenc}
\usepackage[table]{xcolor}
\usepackage{colortbl}
\usepackage{booktabs}
\usepackage{mathtools}
\usepackage{enumitem}
\usepackage{amsmath}
\usepackage{amsthm}
\usepackage{amssymb}
\usepackage{graphicx}
\usepackage{geometry}
\usepackage{setspace}
\usepackage[authoryear]{natbib}
\usepackage[bookmarks=false,
 breaklinks=false,pdfborder={0 0 1},backref=false,colorlinks=true]
 {hyperref}
\hypersetup{
 linkcolor=blue,citecolor=blue,urlcolor=blue}
\allowdisplaybreaks
\makeatletter
\providecommand{\tabularnewline}{\\}
     
\usepackage{color}
\usepackage{amsthm}
\usepackage{graphicx}

\providecommand{\tabularnewline}{\\}

\usepackage{bm}
\usepackage{threeparttable}
\usepackage{caption}
\usepackage{xr}
\newtheorem{assumption}{Assumption}

\newcommand{\E}{\mathbb{E}}
\newcommand{\Prob}{\mathbb{P}}
\newcommand{\R}{\mathbb{R}}
\newcommand{\ind}{\mathbf{1}}

\providecommand{\corollaryname}{Corollary}
\providecommand{\definitionname}{Definition}
\providecommand{\lemmaname}{Lemma}

\providecommand{\remarkname}{Remark}
\providecommand{\theoremname}{Theorem}

\makeatother

\theoremstyle{remark}
\newtheorem{rem}{\protect\remarkname}
\theoremstyle{definition}
\newtheorem{defn}{\protect\definitionname}
\theoremstyle{plain}
\newtheorem{lem}{\protect\lemmaname}
\newtheorem{thm}{\protect\theoremname}

\providecommand{\corollaryname}{Corollary}
\providecommand{\definitionname}{Definition}
\providecommand{\lemmaname}{Lemma}
\providecommand{\remarkname}{Remark}
\providecommand{\theoremname}{Theorem}

\begin{document}
\global\long\def\P{\mathbb{P}}%
\setlength{\abovedisplayskip}{8pt} \setlength{\belowdisplayskip}{8pt} 
\title{Single-Network Finite-Sample Inference\\
 in Strategic Network Formation Models\thanks{We thank Giovanni Compiani, Áureo de Paula, Jiaying Gu, Lixiong
Li, and seminar participants for helpful comments. Claude, ChatGPT and Refine.ink were used for research assistance.}}
\author{Wayne Yuan Gao\thanks{Department of Economics, University of Pennsylvania. Email: \texttt{waynegao@upenn.edu}}
\and Ming Li\thanks{Department of Economics and Risk Management Institute, National University
of Singapore. Email: \texttt{mli@nus.edu.sg}} }
\date{\today}
\maketitle
\begin{abstract}
\noindent We develop a finite-sample valid inference procedure for
strategic network formation models with endogenous network statistics,
using only a single observed network. We impose no restrictions on
network density, equilibrium selection, or the dependence structure
induced by strategic interaction. Using a \emph{bounding-by-$c$}
technique, we obtain realization-wise sandwich inequalities whose
middle term depends only on exogenous covariates and i.i.d.\ pairwise
shocks. These inequalities deliver pathwise identifying restrictions
and simulable finite-sample critical values in both parametric and
semiparametric settings. Our proposed procedure is also computationally
tractable: it does \textit{not }involve solving, simulating, or enumerating
equilibrium (sub)networks, and easily scales to networks with $10{,}000$
agents in simulations. In two empirical applications, we find statistical evidence
for positive link interdependence at the 95\% confidence level.

\bigskip{}
\end{abstract}

\newpage{}

\section{Introduction}\label{sec:intro}

The formation of social and economic networks, such as friendships,
information sharing, trade relationships, R\&D alliances, and interbank
lending, is shaped not only by the characteristics of the agents involved
but also by the agents' strategic considerations over other network
links. A link between two agents might be more attractive when they
share common friends or when it connects them to well-positioned partners,
for various reasons related to information, enforcement, and coordination.
Analyzing network formation while explicitly accounting for strategic interdependence
is a key problem in the econometrics of network formation \citep{de2020econometric}.

The same strategic interdependence that shapes network formation also poses several challenges for econometric inference, particularly when only a single network is observed. First, under pairwise stability
\citep{jackson1996strategic}, the mapping from model primitives to outcome networks  is generally set-valued. Second, endogenous network statistics can
induce complex network dependence. Third, conventional procedures are
often computationally expensive: they often require equilibrium simulation or combinatorial enumeration of subnetwork configurations
\citep[e.g.,][]{miyauchi2016structural,sheng2020structural}.
Equilibrium multiplicity, complex dependence, and computational intractability
are different manifestations of the same endogeneity problem induced by strategic interdependence. 
In particular, the realized network need not have a tractable reduced
form, and statistics formed from its dyads need not satisfy familiar
laws of large numbers or central limit theorems \citep[cf.][]{menzel2021clt}.
Imposing a specific
equilibrium selection rule or weak-dependence approximation could 
restore tractability, but would make inference depend on assumptions
that are difficult to verify from a single realized network.

This paper, to the best of our knowledge, is the first in the econometrics
literature to provide a \emph{finite-sample valid inference} method
on the structural parameters in a \textit{strategic network formation}
model under a \emph{single-network} setting, together with the identification
analysis that underlies it. The model we consider is the following:
\[
Y_{ij}=\ind\left\{ Z_{ij}'\beta_{0}+X_{ij}'\gamma_{0}\geq\varepsilon_{ij}\right\} ,
\]
where $Z_{ij}$ is exogenous, $\varepsilon_{ij}$ is an idiosyncratic
pairwise shock, and $X_{ij}=\phi_{ij}(Y,Z)$ contains observed but
endogenous network statistics, such as the number of common friends.
Since $X_{ij}$ involves realized links in the network, it can depend on every pairwise shock in equilibrium.
The parameter is $\theta_{0}=(\beta_{0}',\gamma_{0}')'$, in which
$\beta_{0}$ governs observed homophily and $\gamma_{0}$ governs
strategic link interdependence. 

The key device that enables us to handle the endogenous $X_{ij}$
is a technique we call \emph{bounding by $c$}. A formed link $Y_{ij}=1$ reveals $\varepsilon_{ij}\leq Z_{ij}'\beta_{0}+X_{ij}'\gamma_{0}$.
If in addition the event that the index is at most a deterministic
scanning threshold $c$, i.e., $Z_{ij}'\beta_{0}+X_{ij}'\gamma_{0}\leq c,$
occurs, transitivity of the inequalities implies $\varepsilon_{ij}\leq c$,
an event involving only the exogenous shock and the number $c$. Similarly, an absent link $Y_{ij}=0$ reveals $\varepsilon_{ij}>Z_{ij}'\beta_{0}+X_{ij}'\gamma_{0}$,
so if in addition the index is at least $c$, we have $\varepsilon_{ij}>c$.
Formed links thus generate lower-bound certificates for the frequency
of the event $\{\varepsilon_{ij}\leq c\}$, while absent links generate
the complementary upper-bound certificates. Averaging
these event implications within cells of exogenous characteristics
sandwiches each cell's empirical error distribution between two observable bounds. Importantly, the middle term of this sandwich depends only on $Z$ and the i.i.d.\ pairwise errors, even though both bounds may depend arbitrarily on the equilibrium network through $X$. Aggregating
over the threshold $c$ converts these elementary bounds into
a family of restrictions across the support
of the latent index. The inequalities hold at every sample size, for
every realization, and at every equilibrium.

A key contribution of this paper is to provide a \emph{valid finite-sample
inference} procedure using these realization-wise bounding-by-$c$
inequalities. At the true parameter, our observable criterion is bounded
by a control statistic involving only the exogenous covariates and
pairwise shocks, and the  distribution of this control
statistic is exactly simulable, based on which we can construct semiparametric
and parametric confidence sets with finite-sample coverage guarantees.
In the parametric setup, the error distribution is specified and  
the control statistic can be simulated directly. We also develop sharper
versions using constrained thresholding and Berk--Jones weighting \citep{berk1979goodness}.
In the semiparametric setup, we show that it suffices to simulate
uniform draws based on the probability integral transform, even though
the error distribution is  unknown. Under our procedure, testing
a candidate parameter only requires dyad-level counting: there is
no equilibrium solving, completion search, or graph-isomorphism matching. Furthermore, the finite-sample coverage
result imposes no restriction on network density, equilibrium selection,
or the dependence induced by strategic interaction. The cost of this
robustness is conservatism: the control statistic bounds
rather than reproduces the sampling distribution of the observable
criterion.

The same inequalities also yield identifying restrictions under large-network
asymptotics, which clarifies the source of identification underlying
our inference approach. Without sparsity or weak dependence assumptions, frequencies
involving $X_{ij}$ need not converge as networks grow large. We therefore
formulate the identifying restrictions through \emph{pathwise limit frequencies}:
the limit superior of the observable lower bound and the limit
inferior of the upper bound, both of which are allowed to remain
random in the large-network limit, reflecting persistent indeterminacy
from global dependence or equilibrium selection. Only the ``middle
sandwiched term'', which consists solely of the exogenous covariates
and i.i.d.\ shocks, converges to a deterministic limit and becomes
the anchoring point for the identifying content. This differs from
the usual population formulation of partial identification, which
begins with deterministic (conditional) probabilities or moments.
Here the observable envelopes may fluctuate indefinitely along the
realized sequence, and (partial) identification is characterized by
restrictions that hold between their pathwise limits. The result provides
the population interpretation for the same inequalities used in our finite-sample
inference procedure.

Our proposed procedures scale to networks of size $n=10{,}000$ in
our simulations. In the baseline simulation design
with a scalar exogenous covariate, the parametric confidence set covers
the truth in every replication (valid and conservative) and certifies the sign of the strategic
coefficient (i.e., reject the null of $\gamma_0=0$ at 95\% confidence level) from $n=400$ onward. The fully semiparametric confidence set is also informative and certifies
the sign in $97.8\%$ of replications at $n=3{,}200$, rising to $99.8\%$ at $n=6{,}400$ and $n=10{,}000$. Additional simulations in the Supplemental Appendix show that confidence sets for $\gamma_0$ maintain coverage with additive common shocks and multivariate exogenous covariates, and that imposing symmetry on the error distribution improves sign certification.

We also apply
our inference procedure to two real-world data sets, a high-school
contact network ($n=327$) and six Twitch friendship networks (up to
$n=9{,}498$). Across all networks in both applications, the projected 95\% confidence sets for the strategic-interdependence coefficient lie strictly above zero, rejecting specifications without strategic interdependence. These findings highlight the importance of accounting for strategic interactions in models of network formation.

\subsubsection*{Related Literature}

The papers most closely related to ours in the econometric literature are \citet{de2018identifying}
and \citet{menzel2026strategic}. \citet{de2018identifying} obtain
partial identification under bounded interaction distance and degree
by classifying agents into local network types, defined by their local
network structures along with node covariates. Their identification
requires covariate and graph-isomorphism matching on local types,
while their computation uses quadratic programs in a continuum-of-agents
formulation. We instead use dyad-level threshold crossing, and impose
no bound on degree, interaction depth, or density. \citet{menzel2026strategic}
also uses dyad-level frequencies, but restricts the endogenous covariates
to a ``node-incident'' form that excludes fundamentally dyadic statistics
such as common friends and derives a large-network Poisson approximation
under additional error tail and equilibrium selection conditions.
Our restrictions allow such dyadic statistics and are exact at the
observed finite $n$. 

The bounding-by-$c$ technique itself originates outside the network
setting: \citet{gao2026identification} introduce it for semiparametric
identification of single-agent dynamic panel binary-choice models under
partial stationarity.\footnote{The underlying event-containment inequality is also related to the generalized instrumental-variable analysis of discrete outcomes
\citep{chesher2010instrumental,chesher2017generalized}. Their bounds utilize the solution map and thus are able to obtain generally sharp identification. The bounding-by-$c$ technique here instead formulates the bounds on the realized endogenous covariates directly without knowledge (or use) of the solution map, which affords tractability in our current strategic network formation setting, though we do not claim sharpness.
} They focus on a panel data setting with cross-sectional independence, while here we consider a very different network formation model where the endogenous covariate can be globally dependent.  In a companion
paper, \citet{gao2026tractable} combine the bounding-by-$c$ approach with subnetwork differencing to accommodate individual fixed effects and focus on the identification of structural parameters. There the differencing
annihilates node-level covariates entering as $Z_{i}+Z_{j}$, and the
identification analysis there relies on population tetrad probabilities,
whose existence in a single large network requires additional conditions
on network dependence and equilibrium selection. We instead work without
fixed effects, allow node-level covariates, and deliver valid finite-sample
inference that invokes no population objects.

For the setting with many independent networks, \citet{sheng2020structural}
studies similar strategic network formation models using bounds on subnetwork
configurations and resolves multiplicity through the 0-1 bounds \citep{tamer2003incomplete}.
Her identifying restrictions also require matching selected subnetwork
configurations, and their asymptotic validity is established as the
number of independent networks grows large. \citet{gualdani2021identification}
instead makes a directed-link game tractable by decomposing it into
local games. Both approaches complement our analysis of a single, potentially large network. Other approaches exploit monotone externalities \citep{miyauchi2016structural},
model equilibrium selection \citep{mele2017dense,badev2021nash},
or use incomplete information to obtain conditional independence and
two-step asymptotics \citep{leung2015twostep,wu2024twostep,ridder2025twostep,comola2026estimating}.
For simultaneous discrete games with incomplete information, \citet{de2012inference}
develop inference on the signs of interaction effects that is robust
to equilibrium multiplicity; their many-market sampling environment
differs from our single pairwise-stable network, but the sign of the
strategic interaction effect is likewise the object our sign-revealing
confidence sets target. Our approach delivers finite-sample inference from a single network
without restricting equilibrium selection or the dependence induced
by strategic interaction.

A separate line of literature considers dyadic network formation
models with unobserved agent heterogeneity but no link interdependence
\citep{graham2017econometric,jochmans2017semiparametric,dzemski2019empirical,gao2020,auerbach2022identification,gao2023logical,zeleneev2026identification,li2026bagging},
where conditional independence of links across dyads can be exploited
to deliver identification and asymptotic inference. Our problem is
different: even conditional on observed and unobserved agent attributes,
$X_{ij}$ is a function of the realized network and can transmit equilibrium
dependence across all dyads. The bounding-by-$c$ technique is 
the leverage that allows us to retain the strategic channel.

Our inference results also contrast with single-network asymptotic
theory based on weak dependence, exchangeability, or density restrictions
\citep{leung2019treatment,kojevnikov2021limit,menzel2021clt,leung2022causal,leung2023network}.
These papers obtain laws of large numbers or central limit theorems
by controlling how strongly distant observations can interact, by
imposing an exchangeable-array representation, and/or by restricting the
density regime. For strategic network formation, \citet{leung2026normal}
derive a central limit theorem under subcritical spillovers and decentralized
equilibrium selection. 
In
contrast to this line of literature on network asymptotics, our inference
procedures are derived directly from finite-sample properties of a
control statistic, thus remaining valid when the global dependence
generated by strategic interaction is present and when local network
statistics do not converge to deterministic limits. 

The closest econometric work on finite-sample inference is \citet{rosen2025finite}.
Their main analysis considers the single-agent semiparametric binary-choice
model under a conditional median restriction, which underlies the
maximum score estimator of \citet{manski1975maximum,manski1985semiparametric},
and their inference likewise exploits a realization-wise inequality
and simulable critical values. Here we consider multi-agent
strategic network formation, where the endogenous network
statistics are jointly determined with the entire equilibrium network
and may therefore depend on all pairwise shocks, and use our bounding-by-$c$
technique to handle the multi-agent strategic interdependence. Another
closely related paper is \citet{li2026finite}, who develop finite-sample
inference in incomplete models with multi-agent interactions. Their
procedure compares an optimal-transport statistic with a simulated
worst-case statistic, obtained by optimizing over the equilibrium
correspondence under a fully parametric specification. In contrast,
our bounding-by-$c$ technique yields closed-form dyad-level inequalities
without simulating, solving, or optimizing over the equilibrium correspondence,
which is prohibitively intractable in our strategic network formation context.
In addition, our inference accommodates both parametric and semiparametric settings.  
\citet{liu2026prediction} use conformal inference under exchangeability
to obtain prediction sets with finite-sample coverage for interval-valued
outcomes. Their target is prediction rather than inference on structural
parameters. \citet{canay2017randomization} 
similarly exploit group invariance in the form of approximate
symmetry in an asymptotic framework, to construct randomization tests
for structural parameters.

Some other inference procedures in network settings concern randomized
or sharp-null causal hypotheses, or tests of hypotheses about the
network distribution itself, rather than structural parameters under
an unknown equilibrium \citep*{athey2018exact,grahampelican2020testing,pelican2022optimal,auerbach2022testing}.
Such procedures exploit a known null randomization, asymptotic approximations
to network test statistics, or a model that simplifies under the no-interaction
null. Our confidence sets, by contrast, cover general values of the strategic interaction parameter and thus yield, by inversion, a test of no strategic interdependence. This test is a special case of the general specification test discussed in Remark \ref{rem:over_id}.

The rest of the paper is organized as follows. Section~\ref{sec:model}
develops the model and identification analysis. Section~\ref{sec:single_inf}
presents the finite-sample inference procedure. Sections~\ref{sec:sim}
and \ref{sec:application} report simulations and applications. We
conclude with Section \ref{sec:con}. Proofs are available in the Appendix.
The Supplemental Appendix provides additional implementation details along with supplemental simulation and theoretical results.

\section{Model and Identification}\label{sec:model}

\subsection{Model Setup and Assumptions}\label{subsec:setup}

Consider a set of $n$ agents indexed by $i=1,\ldots,n$, and an unweighted
network among them represented by an $n\times n$ adjacency matrix
$Y:=(Y_{ij})_{i,j=1}^{n}$, where $Y_{ij}=1$ if a link exists between agents $i$ and $j$,
and $Y_{ij}=0$ otherwise. We focus on undirected networks, i.e.,
$Y_{ij}=Y_{ji}$ for all $i,j$, with $Y_{ii}=0$. We index distinct unordered pairs
as $ij$ with $i<j$, and refer to each such pair $ij$ as a \emph{dyad}.
Throughout the paper every sum $\sum_{ij}$ over dyads is understood
to be $\sum_{i<j}$, so that each dyad is counted exactly once.

We consider the following network formation model with link interdependence
\begin{equation}
Y_{ij}=\ind\left\{ Z_{ij}'\beta_{0}+X_{ij}'\gamma_{0}\geq\varepsilon_{ij}\right\} \label{eq:link_model}
\end{equation}
where $Z_{ij}:=w_{n}(Z_{i},Z_{j})$ is a vector of exogenous covariates
constructed from individual-level exogenous characteristics $Z_{i}$
and $Z_{j}$ via a known symmetric function $w_{n}$,\footnote{We allow $w_{n}$ to depend on the network size $n$, for example through
the rescaling of latent positions \citep[cf.][]{leung2019treatment},
to accommodate sparse designs in large networks.} $X_{ij}$
is a vector of potentially endogenous covariates, which may involve
other links $Y_{hk}$ with $(h,k)\neq(i,j)$, and $\varepsilon_{ij}$
is an idiosyncratic pairwise shock. The exogenous covariates $Z_{ij}$
can include a constant,
homophily effects $|Z_{i}-Z_{j}|$ or $\ind\{Z_{i}=Z_{j}\}$,
level effects $Z_{i}+Z_{j}$, and interactions thereof. Level terms such
as $Z_{i}+Z_{j}$ allow the model to capture differences in linking
propensities associated with observed individual characteristics.\footnote{Our model excludes unrestricted unobserved individual effects, which our companion paper \citet{gao2026tractable} accommodates through subnetwork differencing.} The structural
parameters of interest are $\theta_{0}:=(\beta_{0}',\gamma_{0}')'\in\Theta$,
and throughout this paper we use the shorthand notation
\begin{equation}
\delta_{ij}:=\delta_{ij}(\theta_{0}),\quad\delta_{ij}(\theta):=Z^{'}_{ij}\beta+X^{'}_{ij}\gamma,
\end{equation}
so that the link formation equation reads $Y_{ij}=\ind\{\varepsilon_{ij}\leq\delta_{ij}\}$.

A key feature of our model is the presence of the endogenous covariates
$X_{ij}$, which may arise as functions of the realized network $Y$, i.e.,
\begin{equation}
X_{ij}=\phi_{ij}\left(Y,Z\right)\label{eq:endogenous_X}
\end{equation}
where $\phi_{ij}$ is a known function mapping the realized network $Y$ and the
collection of individual covariate vectors $Z:=(Z_{i})_{i=1}^{n}$ to a vector
of dyadic covariates. We collect these vectors in $X:=(X_{ij})_{1\leq i<j\leq n}$.
We suppose that $\phi_{ij}$ does not depend on the realized own-link status $Y_{ij}$,\footnote{\label{fn:own_link}The own-link exclusion endows model \eqref{eq:link_model}
with the economic interpretation of pairwise stability, which rests
on a \emph{counterfactual} comparison between linking and not linking:
any statistic $\tilde{\phi}_{ij}$ of the full network, such as the
sum of the eigenvector centralities of $i$ and $j$, enters through
the counterfactual difference $\phi_{ij}\left(Y_{-ij},Z\right):=\tilde{\phi}_{ij}\left(\left(Y_{-ij},1\right),Z\right)-\tilde{\phi}_{ij}\left(\left(Y_{-ij},0\right),Z\right)$,
where $\left(Y_{-ij},y_{ij}\right)$ denotes the network with own-link
state $y_{ij}$. If $X_{ij}$ instead depends on the realized $Y_{ij}$,
all our inference results still apply, but \eqref{eq:link_model} loses
its literal pairwise-stability interpretation.} but allow $\phi_{ij}$ to be both locally and globally defined. 

One
canonical local network statistic is the number of common friends
of $i$ and $j$, $\text{CF}_{ij}:=\sum_{k\neq i,j}Y_{ik}Y_{jk}$,
which captures transitivity: the surplus of link $ij$ may increase
with the number of shared neighbors. The normalized version, $\overline{\text{CF}}_{ij}:=\frac{\text{CF}_{ij}}{n-2}\,\in[0,1],$
is exactly the friends-in-common statistic of \citet{sheng2020structural},
and a nonnegative coefficient on it delivers the nonnegative externalities
exploited by \citet{miyauchi2016structural}. Other leading choices of local network statistics are accommodated
without modification: for example, the Jaccard index $J_{ij}:=\text{CF}_{ij}/\sum_{k\neq i,j}\ind\{Y_{ik}+Y_{jk}\geq1\}$
(with $J_{ij}:=0$ when the denominator vanishes), which measures
overlap \emph{relative} to the union of the two neighborhoods; second-order
friends  $\overline{\text{SF}}_{ij}:=\frac{1}{n-2}\sum_{k\neq i,j}\left(Y_{ik}+Y_{jk}\right)\in[0,2]$; 
and interactions with exogenous types, such as $\overline{\text{CF}}_{ij}\cdot\ind\{Z_{i}=Z_{j}\}$,
which let the strength of transitivity vary with observables. 
Notably, our model also accommodates globally defined $X_{ij}$, i.e.,
those that cannot be pinned down completely by the local network structure
around $ij$. Leading examples include various globally defined centrality
measures, such as eigenvector centrality and closeness centrality
(in the differenced form as described in Footnote \ref{fn:own_link}).
Such globally defined $X_{ij}$ are usually highly nonlinear and implicit
aggregations of links from the whole network $Y$, for which there
exist no closed-form representations.

Because $X_{ij}$ may depend on the realized network $Y$, equation~\eqref{eq:link_model} determines links jointly. The resulting network can be interpreted as a pairwise stable equilibrium of a strategic network formation game with transferable utility \citep[cf.][]{jackson1996strategic,sheng2020structural}. In general,
the game may have multiple equilibria for a given realization of primitives
$(Z,\varepsilon)$, where $\varepsilon:=(\varepsilon_{ij})_{1\leq i<j\leq n}$
collects all the idiosyncratic shocks. We represent the realized network
abstractly by 
\begin{equation}
Y\in g(Z,\varepsilon;\theta_{0})\label{eq:equilibrium}
\end{equation}
where $g$ is the equilibrium correspondence. The equilibrium correspondence $g$ generally lacks a tractable characterization. Even simple specifications of $\phi_{ij}$ can generate complex dependence through feedback among link decisions. The discrete, combinatorial space of possible networks further complicates the characterization of $g$. When $\phi_{ij}$ depends on features of the entire network, link decisions may also depend directly on distant links. Our identification approach does not require characterizing, evaluating, or computing $g$. As shown below, we extract identifying information from monotonicity restrictions that hold for every equilibrium, regardless of how complex the equilibrium correspondence is or how an equilibrium is selected.

\medskip{}

We maintain the following assumptions throughout. 

\begin{assumption}[Model Specification] \label{ass:model} 
The observed data consist of one realization of $(Y,Z)$. For some fixed
$\theta_{0}\in\Theta$, equation \eqref{eq:link_model} holds almost surely
for every $1\leq i<j\leq n$.
\end{assumption} 

\begin{assumption}[IID
Pairwise Errors] \label{ass:iid_errors} The components of $\varepsilon$
are independently and identically distributed, with common CDF $F_{\varepsilon}$.
\end{assumption} 

\begin{assumption}[Exogeneity of $Z$] \label{ass:exogeneity}
The random collections $Z$ and $\varepsilon$ are independent.
\end{assumption} 

Assumptions~\ref{ass:model}--\ref{ass:exogeneity}
are standard in the strategic network formation literature. Under the transferable-utility interpretation of \eqref{eq:link_model}, Assumption~\ref{ass:model} requires the observed network to be pairwise stable.\footnote{Assumption~\ref{ass:model}
also fixes the data environment we consider: a \emph{single observed
network}. The alternative environment of many independent networks
(villages, schools, classrooms, markets) is more tractable, since conditional
probabilities of observable events given the agents' exogenous characteristics
are then directly identified and consistently estimable across networks,
with no restriction on within-network dependence. Our restrictions
apply to that case as well (Remark~\ref{rem:many}), but the single-network
environment is the one our procedures are designed for.}
Assumption~\ref{ass:iid_errors} requires the link-level surplus shocks to be independent and identically distributed across dyads but leaves their common distribution $F_{\varepsilon}$ unrestricted at this stage. Our identification strategy applies in both parametric and semiparametric settings. We impose a parametric specification for $F_{\varepsilon}$ only where needed.\footnote{As is standard in binary-outcome
models, the pair $(\theta_{0},F_{\varepsilon})$ is restricted only
up to location and scale: adding a constant to the index and to $\varepsilon$,
or multiplying both by a common $b>0$, leaves \eqref{eq:link_model}
unchanged. A parametric specification such as the standard normal
or logistic fixes both. In the semiparametric case we normalize the
intercept to zero and fix the scale of one nonzero coefficient as
$|\beta_{1}|=1$, implemented as the union of the two sign charts
$\beta_{1}=\pm1$ since a positive rescaling cannot change a coefficient's
sign.} Assumption~\ref{ass:exogeneity}
is an exogeneity condition on the observables $Z$. Importantly, we do \emph{not} assume that $X$
and $\varepsilon$ are independent: since $X$ is a function of the
equilibrium network, it is endogenous and  dependent on 
$\varepsilon$. This assumption also formalizes the naming
convention of calling $Z$ the exogenous covariates while calling
$X$ the endogenous covariates. 

\begin{rem}[Additive Common Shocks]
\label{rem:dependence}An additive shock shared by all dyads can be
absorbed into the free intercept in the parametric model, or into
the unrestricted error location in the semiparametric model. The
structural slope coefficients are unaffected. In the parametric case,
the intercept represents the combined effect of the structural intercept
and the realized common shock, which cannot be separately identified
from a single network. Provided the remaining dyadic shocks satisfy
our maintained assumptions conditional on the common shock, the
finite-sample coverage guarantees apply to these parameters.
Supplemental Appendix S.3.1 illustrates the parametric case and its validity via simulations.
\end{rem}

\subsection{The Bounding-by-$c$ Event Inequalities
}\label{subsec:bbc}

This subsection develops the bounding-by-$c$ inequalities at the
level of individual \emph{events}, holding for every realization of
the primitives, every equilibrium consistent with them, and every
selection among equilibria. They underlie both the identifying restrictions
of Section~\ref{sec:single} and the finite-sample inference theory
of Section~\ref{sec:single_inf}.

Under model \eqref{eq:link_model}, each link satisfies the threshold-crossing
condition $Y_{ij}=\ind\{\varepsilon_{ij}\leq\delta_{ij}\}$.
The key difficulty is that $\delta_{ij}$ contains the endogenous covariate $X_{ij}$. Its equilibrium value depends on the entire shock vector $\varepsilon$ through an intractable equilibrium map $g$ and an unknown selection mechanism.

The core idea of the bounding-by-$c$ technique is to bound comparisons involving an endogenous threshold using a deterministic threshold $c$. Fix any constant $c\in\R$
and consider the event 
\[
Y_{ij}=1\quad\text{and}\quad\delta_{ij}\leq c.
\]
Given that $Y_{ij}=\ind\{\varepsilon_{ij}\leq\delta_{ij}\}$, together
with the event $\left\{ \delta_{ij}\leq c\right\} $, we can deduce
from simple transitivity of the two inequalities that $\varepsilon_{ij}\leq c$,
an event that involves only the exogenous shock and the deterministic
threshold $c$. Similarly, the flipped event ($Y_{ij}=0$ and $\delta_{ij}\geq c$)
implies $\varepsilon_{ij}>\delta_{ij}\geq c$ and thus $\varepsilon_{ij}>c$.
Formally, for any $c\in\R$, 
\begin{align}
Y_{ij}\ind\left\{ \delta_{ij}\leq c\right\} \  & \leq\ \ind\left\{ \varepsilon_{ij}\leq c\right\} \leq1-\left(1-Y_{ij}\right)\ind\left\{ \delta_{ij}\geq c\right\}.\label{eq:dyad_inclusion}
\end{align}
These are our core bounding-by-$c$ sandwich inequalities.

Three remarks on the bounding-by-$c$ technique are in order. First,
the inequalities in \eqref{eq:dyad_inclusion} are valid algebraically
for every realization of the primitives $(Z,\varepsilon)$, the equilibrium
network $Y$, and the induced network statistics $X$. Nothing about
the equilibrium correspondence, its multiplicity, or the selection
mechanism is used beyond model \eqref{eq:link_model} itself. This
is why the resulting restrictions are robust to equilibrium multiplicity
and arbitrary, unknown selection.

Second, nonnegative weighted sums and, more generally, coordinatewise nondecreasing transformations preserve the inequalities in \eqref{eq:dyad_inclusion}. Averaging within covariate cells underpins the inference procedures in Section~\ref{sec:single_inf}. Taking conditional expectations yields analogous restrictions wherever those expectations are identified (Remark~\ref{rem:many}).
The potential complex endogeneity of $X_{ij}$ affects neither the validity, nor the computability, of such aggregations: no conditional distribution
of $X_{ij}$ given $Z$ is ever modeled, estimated, or simulated.

Third, the test constant $c$ generates a one-dimensional family
of restrictions. Each $c$ contributes one lower and one upper bound
on $\ind\left\{ \varepsilon_{ij}\leq c\right\} $, whose probability
is exactly the error CDF $F_{\varepsilon}\left(c\right)$. Hence, sweeping $c$
over $\R$ is effectively a way to trace out bounds for the entire
error CDF. The identifying content
for $\theta$ comes from the interplay between the $c$-family, the
aggregation across dyads, and the conditioning on the dyad's exogenous
characteristics $(Z_{i},Z_{j})$.

\subsection{Identification under Large-Network Asymptotics}\label{sec:single}

This subsection translates the realization-wise inequalities \eqref{eq:dyad_inclusion} into identifying restrictions under large-network
asymptotics, in which the number of agents $n$ grows to infinity. It should be emphasized that \emph{none} of the
finite-sample coverage guarantees of Section \ref{sec:single_inf}
relies on the identification analysis here, but it clarifies where the identifying power
behind the inference procedure comes from.

Formally, we place all primitives on a single probability space: an infinite
sequence $(Z_{i})_{i\geq1}$ of individual covariates, i.i.d.\ by
Assumption~\ref{ass:random_sampling} below, an array $(\varepsilon_{ij})_{1\leq i<j<\infty}$
of i.i.d.\ shocks independent of the covariate sequence, and an auxiliary randomization
variable $\xi$ independent of the two collections jointly (allowing randomized
equilibrium selection). We write $Z^{(n)}=(Z_{i})_{i=1}^{n}$ and
$\varepsilon^{(n)}=(\varepsilon_{ij})_{1\leq i<j\leq n}$ for their finite restrictions.
For each $n$, let $Y^{(n)}=g_{n}(Z^{(n)},\varepsilon^{(n)};\theta_{0},\xi)$ be the realized network
formed from the first $n$ agents by the ($n$-dependent) equilibrium-and-selection
mapping $g_{n}$, which completes the equilibrium correspondence $g$ in \eqref{eq:equilibrium} with an arbitrary equilibrium selection mechanism  under abstract (global) randomness in  $\xi$.  We write $\P$ for the joint
law of $\big((Z_{i})_{i\geq1},(\varepsilon_{ij})_{1\leq i<j<\infty},\xi\big)$.
The resulting equilibrium networks $Y^{(n)}$ need not be nested across $n$.

\subsubsection*{The Finite-Sample Sandwich}

We now formulate the finite-sample conditional averages of
the bounding-by-$c$ inequalities in Section \ref{subsec:bbc}. We
first define a discrete random variable $Z_{D,ij}$ below,
which allows us to cover discrete, continuous, and mixed $Z_{i}$
in a unified manner. 
\begin{defn}[Discretized Dyad Types/Cells]
\label{def:disc} Let $\left({\cal Z}_{D,1},\ldots,{\cal Z}_{D,\kappa}\right)$
be an arbitrary finite, \emph{symmetric}\footnote{Symmetric in the sense that $\left(z,z^{\prime}\right)\in{\cal Z}_{D,k}$
if and only if $\left(z^{\prime},z\right)\in{\cal Z}_{D,k}$.} partition of $\text{Supp}\left(Z_{i},Z_{j}\right)\equiv\text{Supp}\left(Z_{i}\right)^{2}$
such that $\P\left(\left(Z_{i},Z_{j}\right)\in{\cal Z}_{D,k}\right)>0$
for each $k=1,\ldots,\kappa$. We define the discretized dyad type variable 
$Z_{D,ij}$ as $Z_{D,ij}:=z_{D,k}$ if $\left(Z_{i},Z_{j}\right)\in{\cal Z}_{D,k}$,
where $z_{D,k}$ are $\kappa$ distinct arbitrary cell labels. We
write ${\cal Z}_{D}:=\left\{ z_{D,1},\ldots,z_{D,\kappa}\right\} $
for the support of $Z_{D,ij}$. 
\end{defn}

When $Z_{i}$ has continuous or mixed support, the above device allows
us to proceed with a discretized dyad type, with \emph{no loss of
validity}: the finite-sample sandwich \eqref{eq:empirical_sandwich}
below and the coverage guarantees of Section \ref{sec:single_inf}
hold for any partition as in Definition \ref{def:disc}.\footnote{The
discretized type $Z_{D,ij}$ enters only as the conditioning variable
through which dyads are grouped into cells. In the meanwhile, the model itself is \emph{never
discretized}, and $\delta_{ij}(\theta)$ is always evaluated at the
raw covariates. Discretization therefore costs informational coarsening,
not validity: the restrictions hold cell by cell for any partition,
so a coarser partition simply delivers fewer of them. Coarsening can
even be desirable in finite samples, since larger cells make the exogenous
middle term concentrate more tightly and hence yield smaller critical
values. See Section \ref{sec:single_inf} for more details.} When $Z_{i}$ has finite support, one can always take the finest
symmetric partition, whose cells are the diagonal points $\left\{ (z,z)\right\} $
and the symmetrized pairs $\left\{ (z,z^{\prime}),(z^{\prime},z)\right\} $
for $z\neq z^{\prime}$ in $\text{Supp}\left(Z_{i}\right)$. That
said, in computation, one may still want to aggregate support points
in $\text{Supp}\left(Z_{i},Z_{j}\right)$ into coarser cells even
when $Z_{i}$ is discrete. 

We start by fixing a candidate parameter $\theta$, a threshold $c$,
and a generic dyad type $z_{D}$ in the support of $Z_{D,ij}$. Define
the realized cell frequencies 
\begin{align}
\hat{p}^{\,(n)}_{L}(z_{D},c;\theta) & :=\frac{\sum_{ij}Y_{ij}\ind\{\delta_{ij}(\theta)\leq c\}\ind\{Z_{D,ij}=z_{D}\}}{\sum_{ij}\ind\{Z_{D,ij}=z_{D}\}},\label{eq:phatL}\\
\hat{p}^{\,(n)}_{U}(z_{D},c;\theta) & :=1-\frac{\sum_{ij}(1-Y_{ij})\ind\{\delta_{ij}(\theta)\geq c\}\ind\{Z_{D,ij}=z_{D}\}}{\sum_{ij}\ind\{Z_{D,ij}=z_{D}\}},\label{eq:phatU}
\end{align}
together with the \emph{exogenous middle term} 
\begin{equation}
\hat{F}_{n}(c;z_{D}):=\frac{\sum_{ij}\ind\{\varepsilon_{ij}\leq c\}\ind\{Z_{D,ij}=z_{D}\}}{\sum_{ij}\ind\{Z_{D,ij}=z_{D}\}}.\label{eq:Fhat}
\end{equation}
All three objects are defined on the event that $\sum_{ij}\ind\{Z_{D,ij}=z_{D}\}\geq1$,
which occurs almost surely for all large $n$ since $\P\left(Z_{D,ij}=z_{D}\right)=\P\left(\left(Z_{i},Z_{j}\right)\in{\cal Z}_{D,k}\right)>0$.
On the complementary event of an empty cell, we adopt the convention
$(\hat{p}^{\,(n)}_{L},\hat{F}_{n},\hat{p}^{\,(n)}_{U}):=(0,0,1)$,
which preserves the sandwich inequality below and, almost surely,
affects only finitely many terms of each sequence. We emphasize that
$\delta_{ij}(\theta)=Z^{'}_{ij}\beta+X^{'}_{ij}\gamma$ is still defined
with the raw exogenous covariates $Z_{ij}$, \emph{not} recomputed
based on the discretized $Z_{D,ij}$. The frequencies $\hat{p}^{\,(n)}_{L}$ and $\hat{p}^{\,(n)}_{U}$
are computable from observed data $(Y,X,Z)$, while the middle term
$\hat{F}_{n}$ involves the unobserved shocks.

Averaging \eqref{eq:dyad_inclusion} within a given cell $z_D$ therefore yields the following \emph{finite-sample
oracle sandwich}: for every $n$,
\begin{equation}
\hat{p}^{\,(n)}_{L}(z_{D},c;\theta_{0})\ \leq\ \hat{F}_{n}(c;z_{D})\ \leq\ \hat{p}^{\,(n)}_{U}(z_{D},c;\theta_{0})\quad\text{for all }c\in\R\label{eq:empirical_sandwich}
\end{equation}
at every $z_{D}$ such that $\sum_{ij}\ind\{Z_{D,ij}=z_{D}\}\geq1$.
The sandwich inequality \eqref{eq:empirical_sandwich} is stronger
than a population moment inequality: it holds for each realization \emph{before} any expectation
is taken, and under arbitrary dependence in the realized network.
All equilibrium (endogeneity) complications are confined to the two observable outer
terms $\hat{p}^{\,(n)}_{L}$ and $\hat{p}^{\,(n)}_{U}$, while the middle term $\hat{F}_{n}$ contains only the exogenous $Z$ and
$\varepsilon$.

\subsubsection*{A Uniform Strong LLN for the Exogenous Middle Term $\hat{F}_{n}$}

What drives our large-network identification result is the convergence
of the middle term $\hat{F}_{n}(c;z_{D})$ as $n\to\infty$. This
convergence is the only place where we need the agents to be randomly sampled
from a common population, which we now impose. 
\begin{assumption}[Random Sampling]
\label{ass:random_sampling} The individual-level exogenous characteristics
$Z_{i}$ are independently and identically distributed across $i$.
\end{assumption}
Assumption~\ref{ass:random_sampling} is again a standard assumption in the network formation literature: in our paper it is used only in this subsection.
None of the finite-sample inference results in Section \ref{sec:single_inf}
invokes it, where we condition on a finite $n$ and the realized $Z$. 
\begin{lem}[Uniform LLN for the Middle Term]
\label{lem:exo_middle} Under Assumptions~\ref{ass:iid_errors}--\ref{ass:random_sampling},
\[
\max_{z_{D}\in\mathcal{Z}_{D}}\ \sup_{c\in\R}\ \big|\hat{F}_{n}(c;z_{D})-F_{\varepsilon}(c)\big|\ \longrightarrow\ 0,\quad\P\text{-almost surely.}
\]
\end{lem}
\noindent The proof is elementary after observing that the numerator
and denominator of $\hat{F}_{n}$ are $U$-statistics of the jointly
exchangeable, dissociated array formed by i.i.d. $Z_{i}$ and $\varepsilon_{ij}$,
so the strong law of large numbers for exchangeable arrays applies.
The equilibrium and the endogenous $X_{ij}$ play no role at all in this
convergence.

\subsubsection*{Pathwise Identified Set}

We now formulate identifying restrictions in the large-network limit. Define, for each $(z_{D},c,\theta)$, the (possibly random) \textit{pathwise
limit frequencies }
\begin{equation}
p^{\infty}_{L}(z_{D},c;\theta):=\limsup_{n\to\infty}\ \hat{p}^{\,(n)}_{L}(z_{D},c;\theta),\quad p^{\infty}_{U}(z_{D},c;\theta):=\liminf_{n\to\infty}\ \hat{p}^{\,(n)}_{U}(z_{D},c;\theta),\label{eq:pLU_single}
\end{equation}
which always exist. Crucially, the inequality $\limsup_{n}\hat{p}^{\,(n)}_{L}\leq\liminf_{n}\hat{p}^{\,(n)}_{U}$
does \emph{not} follow from the finite-sample sandwich \eqref{eq:empirical_sandwich}
alone, but it does follow once the middle term $\hat{F}_n(c;z_D)$ is shown to \emph{converge}
per Lemma~\ref{lem:exo_middle}: $\P$-almost surely,
\[
\begin{aligned}
p^{\infty}_{L}(z_{D},c;\theta_{0})=&\limsup_n \hat{p}^{\,(n)}_L(z_D,c;\theta_0)
\leq \limsup_n \hat{F}_n(c;z_D) = F_{\varepsilon}(c),\\
p^{\infty}_{U}(z_{D},c;\theta_{0})=&\liminf_n \hat{p}^{\,(n)}_U(z_D,c;\theta_0)
\geq \liminf_n \hat{F}_n(c;z_D) = F_{\varepsilon}(c),
\end{aligned}
\]
which  yields the following.
\begin{thm}
\label{thm:single_pop} Suppose Assumptions~\ref{ass:model}--\ref{ass:random_sampling}
hold. Then, $\P$-almost surely, 
\[
p^{\infty}_{L}(z_{D},c;\theta_{0})\ \leq\ F_{\varepsilon}(c)\ \leq\ p^{\mathrm{\infty}}_{U}(z_{D},c;\theta_{0})\quad\text{for all }c\in\R\text{ and all dyad types }z_{D}.
\]
Consequently, the following results hold almost surely: 
\begin{itemize}
    \item[(i)] Semiparametric case:
\[
\max_{z_{D}\in\mathcal{Z}_{D}}p^{\mathrm{\infty}}_{L}(z_{D},c;\theta_{0})\ \leq\ \min_{z_{D}}p^{\mathrm{\infty}}_{U}(z_{D},c;\theta_{0})\quad\text{for all }c\in\R;
\]
\item[(ii)] Parametric case: if furthermore $F_{\varepsilon}$ is assumed to belong to a
 parametric family $F_{\varepsilon}(\,\cdot\,;\theta_{\varepsilon})$,
then the stacked true parameter $\overline{\theta}_{0}:=(\theta_{0}',\theta_{\varepsilon,0}')'$
satisfies
\[
\max_{z_{D}\in\mathcal{Z}_{D}}p^{\mathrm{\infty}}_{L}(z_{D},c;\theta_{0})\ \leq\ F_{\varepsilon}(c;\theta_{\varepsilon,0})\ \leq\ \min_{z_{D}}p^{\infty}_{U}(z_{D},c;\theta_{0})\quad\text{for all }c\in\R.
\]
\end{itemize}
\end{thm}
\noindent Accordingly, define the semiparametric \emph{pathwise identified set}
\[
\Theta^{\infty}_{I}:=\big\{\theta\in\Theta:\ \max_{z_{D}\in\mathcal{Z}_{D}}p^{\mathrm{\infty}}_{L}(z_{D},c;\theta)\leq\min_{z_{D}}p^{\mathrm{\infty}}_{U}(z_{D},c;\theta)\ \text{for all }c\in\R\big\},
\]
so that Theorem \ref{thm:single_pop}(i) states $\theta_{0}\in\Theta^{\infty}_{I}$
almost surely. The parametric counterpart is defined analogously from
part (ii). The nonstandard formulation of this pathwise identification result warrants some discussion.

First, we clarify that we use the phrase ``identified set'' \emph{without}
claiming sharpness, and there is a substantive reason to doubt that
sharp characterizations are feasible in our context. Sharp identified
sets in games are often characterized by, in effect, \textit{solving
the model}: for each candidate parameter and  realization of the
latent primitives, enumerate the full set of equilibrium outcomes
and check whether some selection reproduces the distribution of the observables. Such sharp identification results have been formalized
through random-set and optimal-transport methods by \citet*{beresteanu2011sharp}
and \citet{galichon2011set}. Here in the specific context of strategic
network formation, the requirement of solving or inverting the pairwise-stability
correspondence over a combinatorial outcome space is intractable in
both theoretical and practical senses, and this intractability is exactly the motivation
behind this paper. The value of the bounding-by-$c$ restrictions
lies in that they extract equilibrium-robust identifying content \emph{without
solving the game}, and inspire the finite-sample inference procedure
in Section~\ref{sec:single_inf}. 

Second, Theorem \ref{thm:single_pop} defines a \textit{valid} identified
set, but does \textit{not} characterize the identified set as a \emph{deterministic}
population quantity as in standard microeconometric settings. Without weak-dependence conditions, $\hat{p}^{\,(n)}_{L}$
and $\hat{p}^{\,(n)}_{U}$ need not converge to deterministic quantities as $n\to \infty$: global interdependence,
or an equilibrium selection that oscillates along the sequence, can
keep them fluctuating, in which case $p^{\infty}_{L}$ and $p^{\infty}_{U}$,
as well as the identified set $\Theta^{\infty}_{I}$, can remain \emph{random}, with such unmodeled randomness abstractly captured by $\xi$.
This marks an important departure from the standard microeconometric identification
analysis, and we regard it as the right one in a single large-network
setting. Conventional identification analysis, point- or set-valued,
rests upon a deterministic feature of the observable distribution
of data: a population conditional probability or moment, or its probability
limit, which the data can in principle recover and upon which the model imposes restrictions.
In our context no such object needs to exist. What is deterministically
pinned down in the limit (or ``population'') is instead the\emph{ middle}
term: the empirical distribution of the exogenous shocks $\hat{F}_{n}$
converges to $F_{\varepsilon}$ along almost every path, and all identifying
restrictions come from sandwiching this convergent middle term between
two possibly random observable bounds. The situation is reminiscent
of laws of large numbers without ergodicity in time-series settings:
for time-series processes, limits of intertemporal averages are
in general random variables in the form of conditional expectations
given the invariant $\sigma$-field rather than constants, and inference
can still proceed given the realized time-series path.
In the same spirit, identification here is defined\emph{ pathwise},
along the realized network sequence, without taking any stand on whether
population limits of observable frequencies exist, or whether there
exists a well-defined limit network model which all finite-$n$ models
converge to. Figure \ref{fig:pathwise} depicts this pathwise formulation
at a fixed $(z_{D},c)$. Along the realized network sequence, global
dependence or an oscillating equilibrium selection can keep the observable
envelopes fluctuating indefinitely, so their pathwise limits may remain
random. Only the exogenous middle term $\hat{F}_{n}$ is guaranteed
to converge, by Lemma \ref{lem:exo_middle}, and it is this convergent
term that anchors the identifying restrictions of Theorem \ref{thm:single_pop}. 

\begin{figure}[tbh]
\begin{centering}
\includegraphics[width=0.62\columnwidth]{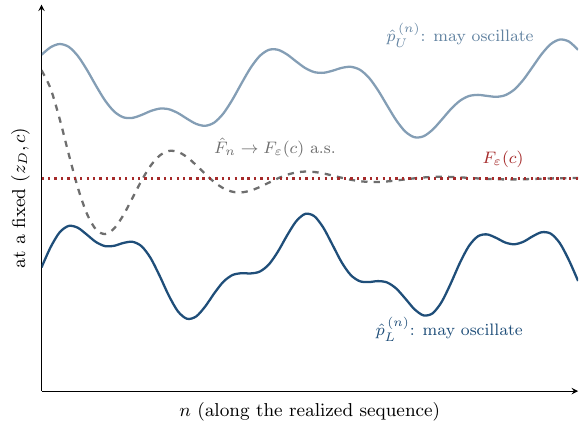}
\par\end{centering}
\caption{Illustration of Pathwise Limits}\label{fig:pathwise}
\end{figure}

\begin{rem}[The Many-Network Environment]
\label{rem:many} When the researcher observes many independent networks
drawn from a common data-generating process, everything above simplifies
and the limit frequencies \eqref{eq:pLU_single} are replaced by ordinary (conditional) 
population moments. Taking expectations of \eqref{eq:dyad_inclusion}
conditional on the dyad type $Z_{D,ij}=z_{D}$, where the middle term
equals $F_{\varepsilon}(c)$ free of $Z$ by Assumptions \ref{ass:iid_errors}
and \ref{ass:exogeneity}, yields the many-network analog of Theorem
\ref{thm:single_pop},
\[
p_{L}(z_{D},c;\theta_{0})\ \leq\ F_{\varepsilon}(c)\ \leq\ p_{U}(z_{D},c;\theta_{0})\quad\text{for every }c\text{ and every dyad type }z_{D},
\]
with $p_{L}(z_{D},c;\theta):=\E[Y_{ij}\ind\{\delta_{ij}(\theta)\leq c\}\mid Z_{D,ij}=z_{D}],$
and $p_{U}(z_{D},c;\theta):=1-\E[(1-Y_{ij})\ind\{\delta_{ij}(\theta)\geq c\}\mid Z_{D,ij}=z_{D}],$
which are directly identified and consistently estimable as the number
of independent networks grows large, with \emph{no} restriction on
within-network dependence. Here, $p_{L}$ and $p_{U}$ are deterministic conditional
probabilities at each given $z_{D}$, as in conventional microeconometric identification
analysis.
\end{rem}

\section{Finite-Sample Inference in a Single Network}

\label{sec:single_inf}

This section develops valid finite-sample inference procedures from
the same realization-wise inequalities that underlie Theorem \ref{thm:single_pop}.
The resulting confidence sets (``CSs''; singular ``CS'') have valid size control at the observed sample size, conditional on the realized
$Z$ with no large-network asymptotics involved.

\subsection{Semiparametric Confidence Sets}

\label{subsec:semi_cs}

The semiparametric identifying restriction in
Theorem~\ref{thm:single_pop}(i) can be encoded by
\[
Q^{\infty}(\theta):=\sup_{c\in\R}\Big[\max_{z_{D}}p^{\infty}_{L}(z_{D},c;\theta)-\min_{z_{D}}p^{\infty}_{U}(z_{D},c;\theta)\Big]_{+},
\]
where $[x]_{+}:=\max\{x,0\}$, so that $Q^{\infty}(\theta_{0})=0$ almost
surely. Our procedure inverts a finite-sample analog of this criterion,
and its logic consists of three steps. First, by the realization-wise
sandwich \eqref{eq:empirical_sandwich}, the observable criterion evaluated
at the true $\theta_{0}$ is bounded by a \emph{control statistic} that
involves only the exogenous $(Z,\varepsilon)$, however complex the
equilibrium network is. Second, by the probability integral transform,
the conditional distribution of the control statistic (given $Z$)  depends only on the cell sizes and not
on the unknown $F_{\varepsilon}$, so it can be simulated exactly from
uniform draws. Third, comparing the observable criterion against a
simulated quantile of the control statistic yields a CS
with valid finite-sample  coverage. We develop each step below.

\subsubsection*{The Test Statistic and the Control Statistic}

The finite-sample analog of $Q^{\infty}$ is
\begin{equation}
\widehat{Q}_{n}(\theta):=\sup_{c\in\R}\Big[\max_{z_{D}\in{\cal \hat{Z}}_{D}}\hat{p}^{\,(n)}_{L}(z_{D},c;\theta)-\min_{z_{D}\in{\cal \hat{Z}}_{D}}\hat{p}^{\,(n)}_{U}(z_{D},c;\theta)\Big]_{+},\label{eq:Qhat_np}
\end{equation}
where $\hat{{\cal Z}}_{D}$ is the family of cells with at least $\underline{m}$
observations, i.e., 
\begin{equation}
\hat{{\cal Z}}_{D}:=\big\{ z_{D}\in{\cal Z}_{D}:\hat{m}(z_{D})\geq\underline{m}\big\},\quad\hat{m}(z_{D}):=\sum_{i<j}\ind\{Z_{D,ij}=z_{D}\},\label{eq:Zhat_mhat}
\end{equation}
for a pre-chosen $\underline{m}\geq1$.\footnote{The choice of the floor $\underline{m}$ involves the following 
trade-off: raising $\underline{m}$ discards restrictions and thereby
lowers the test statistic, but it also removes thin cells $z_D$ that inflate
the simulated critical value constructed below.  In our experience the
second effect typically dominates, so we impose a moderate floor rather
than $\underline{m}=1$.} Validity of our inference procedure is unaffected by any pre-chosen floor,
again by the realization-wise validity of the bounding-by-$c$ inequalities.\footnote{If
no cell reaches the floor, so that $\hat{{\cal Z}}_{D}=\varnothing$,
we define all test and control statistics (and their simulated copies)
to be zero: the CSs below then equal the full candidate
parameter space and the coverage statements hold trivially.} Also, though $\sup_{c\in\R}$ appears to be  taken over
the whole real line, it is actually computable with finitely many
evaluations. This is because, at each fixed $\theta$, the maps $c\mapsto\hat{p}^{\,(n)}_{L}(z_{D},c;\theta)$
and $c\mapsto\hat{p}^{\,(n)}_{U}(z_{D},c;\theta)$ are step functions
whose jumps happen only at the finitely many realized index values
in $\{\delta_{ij}(\theta):i<j\}$, and under the $\leq/\geq$ conventions
in \eqref{eq:phatL}--\eqref{eq:phatU} the supremum of the violation
is attained at one of these values. Hence, $\widehat{Q}_{n}$ is \emph{exactly}
computable at each candidate $\theta$.

To control the finite-sample uncertainty in $\widehat{Q}_{n}(\theta_{0})$,
we use the ``control statistic'' 
\begin{equation}
R_{n}:=\sup_{c\in\R}\Big[\max_{z_{D}\in{\cal \hat{Z}}_{D}}\hat{F}_{n}(c;z_{D})-\min_{z_{D}\in{\cal \hat{Z}}_{D}}\hat{F}_{n}(c;z_{D})\Big],\label{eq:Rn_range}
\end{equation}
which bounds $\widehat{Q}_{n}(\theta_{0})$ by the following lemma. 
\begin{lem}[Semiparametric Uncertainty Control]
\label{lem:dom_semi} For every $n$, $\widehat{Q}_{n}(\theta_{0})\ \leq\ R_{n}.$
\end{lem}
At the true parameter $\theta_{0}$, the statistic $\widehat{Q}_{n}(\theta_{0})$
on the left is observable and may carry arbitrary network endogeneity
and dependence, while the control statistic $R_{n}$ on the right
contains no $Y$ or $X$, and only involves the exogenous $\varepsilon$
and $Z$. Any probability statement about $R_{n}$ based on the independence
assumption between $\varepsilon$ and $Z$ therefore translates directly
into coverage guarantee for the truth $\theta_{0}$.

\subsubsection*{Simulability of the Control Statistic}

It remains to obtain a critical value based on the control statistic $R_{n}$, whose definition
still involves the unobserved shocks with unknown distribution $F_{\varepsilon}$, which we resolve using 
the probability integral transform. Conditional
on $Z$, distinct cells contain disjoint sets of i.i.d.\ shocks (by Assumption
\ref{ass:iid_errors}). Let $U_{ij}:=F_{\varepsilon}(\varepsilon_{ij})$
denote the transformed shocks and, for each cell, let $G_{z_{D}}$
be their empirical CDF,
\[
G_{z_{D}}(u):=\frac{1}{\hat{m}(z_{D})}\sum_{i<j}\ind\{U_{ij}\leq u\}\ind\{Z_{D,ij}=z_{D}\},\quad u\in[0,1].
\]
When $F_{\varepsilon}$ is continuous, the middle term \eqref{eq:Fhat} satisfies, almost surely,
$\hat{F}_{n}(c;z_{D})=G_{z_{D}}\big(F_{\varepsilon}(c)\big)$ for all
$c$: the transform maps each threshold $c$ to the point $u=F_{\varepsilon}(c)$
on the uniform scale $[0,1]$. Substituting into \eqref{eq:Rn_range},
every value that the objective of $R_{n}$ attains as $c$ sweeps $\R$
is also attained by the same objective in $u$, so
\[
R_{n}\ \leq\ R^{\ast}:=\sup_{u\in[0,1]}\Big[\max_{z_{D}\in\hat{{\cal Z}}_{D}}G_{z_{D}}(u)-\min_{z_{D}\in\hat{{\cal Z}}_{D}}G_{z_{D}}(u)\Big],
\]
which holds with exact equality when $F_{\varepsilon}$ is continuous, as shown in the
proof of Theorem \ref{thm:semi_CS}. The key observation is that the
$U_{ij}$ are i.i.d.\ $\mathrm{Uniform}[0,1]$ by the probability integral
transform (again for continuous $F_{\varepsilon}$), so, conditional
on $Z$, the distribution of $R^{\ast}$ depends on the primitives \emph{only
through the cell sizes} $\{\hat{m}(z_{D})\}$: it is the distribution
of the same range statistic computed from independent uniform samples
of those sizes, whatever the unknown $F_{\varepsilon}$ may be. When $F_{\varepsilon}$ has atoms, the randomized distributional transform still yields i.i.d. uniform variables and preserves the almost-sure bound $R_n\leq R^\ast$. The same uniform simulation therefore remains valid, although the bound may be strict and introduce additional conservatism.

We simulate the exact conditional distribution of $R^{\ast}$
as follows. For $b=1,\dots,B$, draw $\hat{m}(z_{D})$ i.i.d.\ uniform
random variables on $[0,1]$ for each $z_{D}\in\hat{{\cal Z}}_{D}$,
independently across $b$ and across cells, and independently of the
observed data and of the realized shocks. Let $G^{(b)}_{z_{D}}$ be
the empirical distribution function of the draws in cell $z_{D}$,
the simulated counterpart of $G_{z_{D}}$, and set
\[
R^{(b)}:=\sup_{u\in[0,1]}\Big[\max_{z_{D}\in\hat{{\cal Z}}_{D}}G^{(b)}_{z_{D}}(u)-\min_{z_{D}\in\hat{{\cal Z}}_{D}}G^{(b)}_{z_{D}}(u)\Big].
\]
Conditional on $Z$, $R^\ast,R^{(1)},\ldots,R^{(B)}$ are i.i.d. Define $\hat{q}^{\,R}_{n,1-\alpha}(Z)$ as the $k$-th smallest value among $R^{(1)},\ldots,R^{(B)}$, where $k:=\lceil(1-\alpha)(B+1)\rceil$. Exchangeability ensures conditional size control at level $\alpha$ for every finite $B$ with $k\le B$.\footnote{\label{fn:finite-B} This is the classical Monte Carlo test construction of \citet{dwass1957modified} and \citet{barnard1963discussion}, and is also discussed in  \citet{dufour2006monte}. The requirement $\lceil(1-\alpha)(B+1)\rceil\leq B$ is mild: at $\alpha=0.05$ the smallest admissible value is $B=19$. We use $B=999$ in our simulations and empirical applications.}

\subsubsection*{The Semiparametric Confidence Set and its Validity}

We then construct the semiparametric CS as
\begin{equation}
\widehat{{\cal C}}^{\mathrm{semi}}_{n}(1-\alpha):=\big\{\theta:\ \widehat{Q}_{n}(\theta)\leq\hat{q}^{\,R}_{n,1-\alpha}(Z)\big\}.\label{eq:CS_semi}
\end{equation}
Probability statements for CSs based on simulated critical
values are understood to be taken over the Monte Carlo draws used
to construct those critical values as well. Note that the inference
procedure uses only the cell sizes $\hat{m}(z_{D})$, \emph{without}
the need to draw $\varepsilon$ from the unknown distribution $F_{\varepsilon}$
or to simulate any network formation process. The power of the test
comes from cross-cell contradictions: because every cell's shocks
are drawn from the same $F_{\varepsilon}$, at $\theta_{0}$ the cell
envelopes must overlap up to oracle sampling noise, whereas an incompatible
$\theta$ opens cross-cell gaps too large to be oracle noise, and
the simulated critical value calibrates exactly this noise scale from
the cell sizes $\hat{m}(z_{D})$. 

\begin{thm}[Semiparametric Confidence Set]
\label{thm:semi_CS} Suppose Assumptions~\ref{ass:model}--\ref{ass:exogeneity}
hold. Then, for every $n$ and every $B$ such that $\lceil(1-\alpha)(B+1)\rceil\leq B$,
\[
\Prob\big(\theta_{0}\in\widehat{{\cal C}}^{\mathrm{semi}}_{n}(1-\alpha)\mid Z\big)\geq1-\alpha.
\]
\end{thm}
The coverage guarantee also holds unconditionally. Theorem \ref{thm:semi_CS},
like every other finite-sample result below, conditions on the realized
$Z$ and therefore does not invoke Assumption \ref{ass:random_sampling}. The proof  formalizes the domination and exchangeability arguments developed
above. An analytic alternative to the simulated critical value, based
on the Dvoretzky--Kiefer--Wolfowitz inequality, together with the rate
at which the critical value shrinks, is provided in Supplemental Appendix 
S.4 (Proposition S.1).

\subsection{Parametric Confidence Sets}

\label{subsec:par_cs}

We now consider the parametric case, in which the true error CDF is
$F_{\varepsilon}(\cdot\,;\theta_{\varepsilon,0})$ for an unknown finite-dimensional
parameter $\theta_{\varepsilon,0}$. We write $\overline{\theta}:=(\theta',\theta_{\varepsilon}')'$
for a candidate stacked parameter, as in Theorem \ref{thm:single_pop}(ii).
The logic of Subsection \ref{subsec:semi_cs} carries over, with one simplification
that shapes the whole construction: the middle term of the sandwich
has a known functional form. At each candidate $\overline{\theta}$, we compare
the observable envelopes directly with $F_{\varepsilon}(\cdot\,;\theta_{\varepsilon})$
and simulate the control statistic under the corresponding null by drawing
shocks from that distribution.

Note that, when $F_{\varepsilon}$
is taken to be normal or logistic (as commonly adopted in the literature),
it is without loss of generality to set $F_{\varepsilon}$ to be standard
normal or standard logistic as location and scale normalization, so
that no free parameter $\theta_{\varepsilon}$ is required.\footnote{When $F_{\varepsilon}$ is set to be standard normal or standard logistic,
a constant term should be included in $Z_{ij}$, and all components
of $\theta$ should be left unrestricted.} That said, we keep  $\theta_{\varepsilon}$ for theoretical
generality.

Specifically, we define the finite-sample test statistic
\begin{equation}
\widehat{Q}^{\mathrm{par}}_{n}(\overline{\theta}):=\sup_{c\in\R}\!\Big[\max\Big\{\!\max_{z_{D}\in\hat{{\cal Z}}_{D}}\hat{p}^{\,(n)}_{L}(z_{D},c;\theta)-F_{\varepsilon}(c;\theta_{\varepsilon}),\,F_{\varepsilon}(c;\theta_{\varepsilon})-\!\min_{z_{D}\in\hat{{\cal Z}}_{D}}\hat{p}^{\,(n)}_{U}(z_{D},c;\theta)\Big\}\Big]_{+}.\label{eq:Qhat_p}
\end{equation}
Correspondingly, we define the control statistic 
\begin{equation}
R^{\mathrm{par}}_{n}\left(\theta_{\varepsilon}\right):=\max_{z_{D}\in\hat{{\cal Z}}_{D}}\ \sup_{c\in\R}\ \big|\hat{F}_{n}(c;z_{D})-F_{\varepsilon}(c;\theta_{\varepsilon})\big|,\label{eq:Rnpar}
\end{equation}
which controls the uncertainty in $\widehat{Q}^{\mathrm{par}}_{n}(\overline{\theta})$
at the truth. 
\begin{lem}[Parametric Uncertainty Control]
\label{lem:dom_par} Suppose $F_{\varepsilon}=F_{\varepsilon}(\cdot\,;\theta_{\varepsilon,0})$
is known up to $\theta_{\varepsilon}$. Then, for every $n$, $\widehat{Q}^{\mathrm{par}}_{n}(\overline{\theta}_{0})\ \leq\ R^{\mathrm{par}}_{n}\left(\theta_{\varepsilon,0}\right).$
\end{lem}
For each candidate $\theta_{\varepsilon}$, we simulate the corresponding null
distribution of $R^{\mathrm{par}}_{n}\left(\theta_{\varepsilon}\right)$
to obtain the critical value. Specifically, for $b=1,\ldots,B$, draw (independently across $b$, and independently
of the observed data) an i.i.d.\ pair-shock
array $\{\varepsilon^{(b)}_{ij}\}_{ij}$ from $F_{\varepsilon}(\cdot\,;\theta_{\varepsilon})$,
compute $\hat{F}^{(b)}_{n}(c;z_{D})$ and then 
\[
R^{\mathrm{par},(b)}_{n}(\theta_{\varepsilon}):=\max_{z_{D}\in\hat{{\cal Z}}_{D}}\sup_{c\in\R}\big|\hat{F}^{(b)}_{n}(c;z_{D})-F_{\varepsilon}(c;\theta_{\varepsilon})\big|.
\]
Let $\hat{q}_{n,1-\alpha}(Z;\theta_{\varepsilon})$ be the $\lceil(1-\alpha)(B+1)\rceil$-th
smallest value of the simulated statistics $R^{\mathrm{par},(b)}_{n}(\theta_{\varepsilon})$ where
$b=1,\ldots,B$ similarly as before. We define the parametric CS as
\[
\widehat{\mathcal{C}}^{\mathrm{par}}_{n}(1-\alpha):=\big\{\overline{\theta}:\ \widehat{Q}^{\mathrm{par}}_{n}(\overline{\theta})\leq\hat{q}_{n,1-\alpha}(Z;\theta_{\varepsilon})\big\}.
\]

\begin{thm}[Parametric Confidence Set]
\label{thm:para_CS} Suppose Assumptions~\ref{ass:model}--\ref{ass:exogeneity}
hold with $F_{\varepsilon}=F_{\varepsilon}(\cdot\,;\theta_{\varepsilon,0})$
known up to the finite-dimensional parameter $\theta_{\varepsilon,0}$.
Then, for every $n$ and every $B$ such that $\lceil(1-\alpha)(B+1)\rceil\leq B$,
\[
\Prob\left(\overline{\theta}_{0}\in\widehat{\mathcal{C}}^{\mathrm{par}}_{n}(1-\alpha)\,\big|\,Z\right)\geq1-\alpha.
\]
\end{thm}
Theorem \ref{thm:para_CS} parallels Theorem \ref{thm:semi_CS}, with
two differences worth noting. First, knowledge of $F_{\varepsilon}$
strengthens the restrictions: the parametric criterion \eqref{eq:Qhat_p}
tests every cell against $F_{\varepsilon}$ two-sidedly, whereas the
semiparametric criterion \eqref{eq:Qhat_np} can exploit only cross-cell
contradictions between the lower and upper bounds. Correspondingly, the parametric CSs
are tighter in the simulations in Section \ref{sec:sim}.
Second, the critical value is simulated by drawing shocks from the
specified $F_{\varepsilon}(\cdot\,;\theta_{\varepsilon})$ rather than
from the uniform distribution, and it therefore varies with the candidate
$\theta_{\varepsilon}$, though not with $\theta$.

\begin{rem}[Overidentification and Specification Testing]\label{rem:over_id}
Set identification does not preclude specification testing. Let
$\widehat{\mathcal{C}}_{n}$ denote a $1-\alpha$ confidence set from
this section, constructed under the conditions of its coverage theorem,
and fix a candidate parameter set $\Theta_{\mathrm{cand}}$.
Consider the null that the model and its maintained assumptions,
including any restriction on the error distribution, hold for some
parameter in $\Theta_{\mathrm{cand}}$. Reject this null when
$\widehat{\mathcal{C}}_{n}\cap\Theta_{\mathrm{cand}}=\varnothing$.
Under the null, the true parameter belongs to $\Theta_{\mathrm{cand}}$,
so an empty intersection excludes it from the confidence set.
The corresponding coverage guarantee therefore gives
\[
\Prob\!\left(\widehat{\mathcal{C}}_{n}\cap\Theta_{\mathrm{cand}}=\varnothing\,\middle|\,Z\right)\leq\alpha.
\]
Taking $\Theta_{\mathrm{cand}}$ to be the full parameter space tests
the model itself; a smaller set tests the model jointly with the imposed
parameter restrictions.
\end{rem}

\subsection{Improved Confidence Sets}\label{subsec:gendom}

The sandwich \eqref{eq:empirical_sandwich} is a \emph{collection}
of pointwise inequalities, one per cell $z_{D}$ and threshold $c$,
so any transformation of both sides that preserves the inequalities
yields another valid test statistic. This freedom lets us account
for the different amounts of finite-sample uncertainty carried
by different $(z_{D},c)$ configurations, which leads to a less conservative
test and hence a tighter CS. Here we focus on the \emph{parametric
case}, and leave the semiparametric counterpart to Supplemental
Appendix S.1.\footnote{\label{fn:kappahat}The refinement used in the
semiparametric columns of our simulations studentizes each cross-cell gap
$\hat{p}^{\,(n)}_{L}(z_{D},c;\theta)-\hat{p}^{\,(n)}_{U}(z^{\prime}_{D},c;\theta)$
by $e(z_{D})+e(z^{\prime}_{D})$, where $e(z_{D})\propto\hat{m}(z_{D})^{-1/2}$
is the cell's Dvoretzky--Kiefer--Wolfowitz radius. The resulting range is
compared with its own $\lceil(1-\alpha)(B+1)\rceil$-th smallest simulated
value from the uniform draws of Section~\ref{subsec:semi_cs}, which
discounts gaps involving thin cells and prevents one thin cell's sampling
noise from setting the allowance for every pair.} 

At a candidate $\overline{\theta}$, define
\begin{align*}
V^{+}(z_{D},c;\overline{\theta}) & :=\big[\hat{p}^{\,(n)}_{L}(z_{D},c;\theta)-F_{\varepsilon}(c;\theta_{\varepsilon})\big]_{+}, & R^{+}(z_{D},c;\theta_{\varepsilon}) & :=\big[\hat{F}_{n}(c;z_{D})-F_{\varepsilon}(c;\theta_{\varepsilon})\big]_{+},\\
V^{-}(z_{D},c;\overline{\theta}) & :=\big[F_{\varepsilon}(c;\theta_{\varepsilon})-\hat{p}^{\,(n)}_{U}(z_{D},c;\theta)\big]_{+}, & R^{-}(z_{D},c;\theta_{\varepsilon}) & :=\big[F_{\varepsilon}(c;\theta_{\varepsilon})-\hat{F}_{n}(c;z_{D})\big]_{+}.
\end{align*}
At the truth $\overline{\theta}_{0}$, the sandwich \eqref{eq:empirical_sandwich}
can be written as $V^{\pm}(z_{D},c;\overline{\theta}_{0})\leq R^{\pm}(z_{D},c;\theta_{\varepsilon,0})$
for every $(z_{D},c,\pm)$, which we abbreviate to $V\left(\overline{\theta}_{0}\right)\leq R\left(\theta_{\varepsilon,0}\right)$
with $V,R$ understood as the concatenations across all such combinations.
Hence, for any  \emph{nondecreasing}\footnote{That is, $T\left(v\right)\leq T\left(v^{'}\right)$ whenever $v\leq v^{'}$
elementwise.}\emph{ aggregation} function $T$ mapping the space of $V$ or $R$
into a scalar, we have
\begin{equation}
T\left(V\left(\overline{\theta}_{0}\right)\right)\ \leq\ T\left(R\left(\theta_{\varepsilon,0}\right)\right).\label{eq:gen_dom}
\end{equation}
Taking $T$ to be the maximum over all $(z_{D},c,\pm)$ combinations
recovers $\widehat{Q}^{\mathrm{par}}_{n}$ and $R^{\mathrm{par}}_{n}$
from Section \ref{subsec:par_cs}, in which case \eqref{eq:gen_dom}
specializes to Lemma \ref{lem:dom_par}. More generally, fix any \emph{known}
measurable nondecreasing $T_{Z,\overline{\theta}}$, which may depend
freely on $\left(Z,\overline{\theta}\right)$ but not on $(Y,X)$,
the realized shocks, or the Monte Carlo draws. 
Simulating $B$ i.i.d.\ copies $R^{(1)},\dots,R^{(B)}$ from $Z$ and
$F_{\varepsilon}\left(\,\cdot\,;\theta_{\varepsilon}\right)$ exactly
as in Section \ref{subsec:par_cs}, and letting $\hat{c}^{\,T}_{n,1-\alpha}(Z;\overline{\theta})$
be the $\lceil(1-\alpha)(B+1)\rceil$-th smallest value of the aggregated
draws $T_{Z,\overline{\theta}}(R^{(b)})$, the $T$-induced CS below is again finite-sample valid:
\[
\widehat{\mathcal{C}}^{T}_{n}(1-\alpha):=\left\{ \overline{\theta}:\ T_{Z,\overline{\theta}}\left(V\left(\overline{\theta}\right)\right)\leq\hat{c}^{\,T}_{n,1-\alpha}\left(Z;\overline{\theta}\right)\right\} .
\]

\begin{thm}[Valid Parametric Confidence Sets under Nondecreasing Aggregation]
\label{prop:gendom} Suppose Assumptions~\ref{ass:model}--\ref{ass:exogeneity}
hold with $F_{\varepsilon}=F_{\varepsilon}(\cdot\,;\theta_{\varepsilon,0})$
known up to the finite-dimensional parameter $\theta_{\varepsilon,0}$.
Then, for every $n$ and every $B$ such that $\lceil(1-\alpha)(B+1)\rceil\leq B$,
\[
\Prob\left(\overline{\theta}_{0}\in\widehat{\mathcal{C}}^{T}_{n}(1-\alpha)\,\big|\,Z\right)\geq1-\alpha.
\]
\end{thm}
The choice of $T$ matters because it governs how uncertainty is weighted
across $(z_{D},c,\pm)$, and hence the distribution of $T(R(\theta_{\varepsilon,0}))-T\big(V(\overline{\theta}_{0})\big)$.
The maximum used in Lemma \ref{lem:dom_par} is one of the simplest and most intuitive encoding
of the identifying restrictions in Theorem \ref{thm:single_pop}, but in finite samples it can suffer from high uncertainty: one thin cell at an uninformative threshold
can become the maximum by chance and move the test statistic or the
critical value drastically. A better aggregation limits the influence
of any single $(z_{D},c,\pm)$ configuration that only has very few data
points.

Throughout Sections~\ref{sec:sim} and~\ref{sec:application}, we implement the following two aggregation improvements: \emph{constrained thresholding} and \emph{Berk--Jones weighting}. We provide a short description below to convey the key ideas underlying these improvements, while more details of the exact implementation are available in Supplemental
Appendix S.1. 

First, we work with a restricted set of test thresholds $c$, which we call \emph{constrained thresholding}. At each $(\theta,Z)$ we keep
only those $c$ lying in the \emph{a priori} range of the candidate
index in cell $z_{D}$, 
\[
C(z_{D};\theta,Z)=\Big[\inf_{\tilde{z}:\,\tilde{z}_{D}=z_{D}}w_{n}\left(\tilde{z}\right)^{'}\beta+\inf_{x\in{\cal X}}\gamma^{\prime}x,\ \ \sup_{\tilde{z}:\,\tilde{z}_{D}=z_{D}}w_{n}\left(\tilde{z}\right)^{'}\beta+\sup_{x\in{\cal X}}\gamma^{\prime}x\Big],
\]
where ${\cal X}$ is a known deterministic set containing every value
$X_{ij}$ can take, so that the realized $X$ never enters and $C(z_{D};\theta,Z)$
is $(\theta,Z)$-measurable. Outside this window neither $\hat{p}^{\,(n)}_{L}$
nor $\hat{p}^{\,(n)}_{U}$ varies with $c$, so those coordinates say
nothing about $\theta$, while their oracle counterparts $R^{\pm}$
keep inflating the simulated maximum. Discarding them therefore leaves
the observable side unchanged and weakly lowers the critical value.

Second, we apply \emph{Berk--Jones weighting} \citep{berk1979goodness}.
Since the shocks are independent, their count below a threshold is
binomial within each cell. This motivates the score 
$\hat{m}(z_{D})\mathrm{KL}(p\|q)$, which accounts for both cell size
and the candidate probability, where
\[
\mathrm{KL}(p\|q)=p\log(p/q)+(1-p)\log((1-p)/(1-q)).
\]
The idea is to measure the discrepancy between two probabilities on a
common scale reflecting tail probabilities, making discrepancies
comparable across cells. To score only violations of the observable
bounds, define
\[
d_{\mathrm{BJ}}(a,b;q):=
\begin{cases}
\mathrm{KL}(a\|q),&q<a,\\[3pt]
\mathrm{KL}(b\|q),&q>b,\\[3pt]
0,&a\leq q\leq b.
\end{cases}
\]
The criterion is then
\[
\widehat{Q}^{\mathrm{BJ}}_{n}(\overline{\theta})
=\max_{z_{D}\in\hat{{\cal Z}}_{D}}\ \sup_{c\in C(z_{D};\theta,Z)}\hat{m}(z_{D})\,
d_{\mathrm{BJ}}\!\left(\hat{p}^{\,(n)}_{L},\hat{p}^{\,(n)}_{U};F_{\varepsilon}(c;\theta_{\varepsilon})\right),
\]
with the bounds evaluated at $(z_{D},c;\theta)$. The matching control
statistic replaces the observable bounds with the empirical shock
frequency:
\[
R^{\mathrm{BJ}}_{n}(\overline{\theta})
=\max_{z_{D}\in\hat{{\cal Z}}_{D}}\ \sup_{c\in C(z_{D};\theta,Z)}\hat{m}(z_{D})\,
\mathrm{KL}\!\left(\hat{F}_{n}(c;z_{D})\middle\|F_{\varepsilon}(c;\theta_{\varepsilon})\right).
\]
At the truth, the empirical shock frequency lies between the bounds,
so $\widehat{Q}^{\mathrm{BJ}}_{n}\leq R^{\mathrm{BJ}}_{n}$.
Simulating the control statistic under the candidate error distribution
supplies the critical value, preserving finite-sample coverage by
Theorem~\ref{prop:gendom}. Supplemental Appendix S.1 
provides the implementation details.

\section{Simulation Study}

\label{sec:sim}

This section evaluates the performance of the CSs of
Section~\ref{sec:single_inf} across network sizes from $100$ to $10{,}000$. 
Throughout, every reported CS is computed from a \emph{single} simulated network.
Three further exercises are presented in Supplemental Appendix S.3:
common and dyadic shock dependence, multidimensional exogenous covariates,
and the symmetry-restricted semiparametric procedure.

\subsection{Design}

\label{subsec:sim_design}

We consider the following model specification
\begin{equation}
Y_{ij}=\ind\Big\{\beta_{0}\,|z_{i}-z_{j}|+\gamma_{0}\,\overline{\mathrm{CF}}_{ij}(Y)-b_{0}\ \geq\ \varepsilon_{ij}\Big\},\label{eq:sim_dgp}
\end{equation}
with true coefficients $\beta_{0}=-1$ and $\gamma_{0}=16$, $\varepsilon_{ij}$
drawn i.i.d.\ from\ $\mathrm{Logistic}(0,1)$,
$z_{i}$ drawn i.i.d.\ uniformly on a grid of $21$ equally spaced
points in $[-10,10]$, and $\overline{\mathrm{CF}}_{ij}(Y):=\mathrm{CF}_{ij}(Y)/(n-2)$
being the normalized common-friend count on the equilibrium network
$Y$. This normalization ensures that the data-generating process,
and thus the equilibrium network, stays stable across $n$. Conditioning
cells ($z_{D}\in{\cal Z}_{D}$) are the $231$ exact unordered type
pairs. We consider network size $n\in\{100,200,400,800,1{,}600,3{,}200,6{,}400,10{,}000\}$
with $K=500$ independent replications at every size. We calibrate
$b_{0}$ via pilot simulation runs, on seeds disjoint from the $K=500$
production seeds, so that the outcome networks at
$n=400$ have densities around 20\%, which yields $b_{0}=-0.6$, fixed
thereafter.

The choice of the strategic coefficient $\gamma_{0}=16$ seems ``large'',
but its magnitude should be interpreted together with the scale of
$X_{ij}:=\mathrm{CF}_{ij}(Y)/(n-2)$, which is normalized by $1/(n-2)$.
Table \ref{tab:sim_design} reports summary statistics about the equilibrium
networks, showing that the endogenous covariate $X$ is generally
small in scale, with mean ranging from $0.041$ to $0.054$ and sd from
$0.054$ to $0.085$ across network sizes. Hence, $\gamma_{0}=16$
only translates to a very moderate amount of variation in the strategic
index term. For example, at $n=200$ the strategic term contributes
$\gamma_{0}\cdot\mathrm{mean}(X)\approx0.80$ logit units at the mean
and $\gamma_{0}\cdot\mathrm{sd}(X)\approx1.20$ logit units of variation,
against a standard logistic shock $\varepsilon_{ij}$. Equivalently,
since $X_{ij}=\mathrm{CF}_{ij}/(n-2)$, at $n=400$ the value $\gamma_{0}=16$
corresponds to an index increment of $\gamma_{0}/(n-2)\approx0.04$ per
raw common friend, comparable to the per-common-friend effects found
in the empirical applications of Section \ref{sec:application}.

\begin{table}[t]
\centering
\setlength{\tabcolsep}{4pt}
 \caption{Summary statistics about the simulated networks}
\label{tab:sim_design} \begin{threeparttable}%
\begin{tabular}{@{}lccccc@{}}
\toprule
\multicolumn{1}{@{}l}{$n$} & \multicolumn{1}{c}{density mean (\%)} & \multicolumn{1}{c}{density range (\%)} & \multicolumn{1}{c}{mean $X$} & \multicolumn{1}{c}{sd $X$} & \multicolumn{1}{c@{}}{mean $\max X$}\tabularnewline
\midrule 
$100$ & $20.40$ & $[12.93,54.73]$ & $0.054$ & $0.085$ & $0.324$\tabularnewline
$200$ & $20.56$ & $[14.30,35.55]$ & $0.050$ & $0.075$ & $0.306$\tabularnewline
$400$ & $20.44$ & $[16.17,31.05]$ & $0.047$ & $0.067$ & $0.281$\tabularnewline
$800$ & $20.21$ & $[17.58,27.38]$ & $0.044$ & $0.061$ & $0.254$\tabularnewline
$1{,}600$ & $20.06$ & $[18.26,23.57]$ & $0.042$ & $0.057$ & $0.228$\tabularnewline
$3{,}200$ & $19.97$ & $[18.74,22.01]$ & $0.042$ & $0.055$ & $0.211$\tabularnewline
$6{,}400$ & $19.95$ & $[18.99,21.19]$ & $0.041$ & $0.054$ & $0.197$\tabularnewline
$10{,}000$ & $19.93$ & $[19.18,20.81]$ & $0.041$ & $0.054$ & $0.190$\tabularnewline
\bottomrule
\end{tabular}\begin{tablenotes}[flushleft]
\raggedright

{\footnotesize\item \emph{Notes:} $X=\overline{\mathrm{CF}}$. ``mean
$\max X$'' averages each network's maximum across replications.
}\end{tablenotes} \end{threeparttable}
\end{table}

Since $\overline{\mathrm{CF}}$ is nondecreasing
in the network and $\gamma_{0}>0$, the network formation game is
supermodular, and the equilibrium networks are computed as the least
fixed points by iterations from the empty graph. While our inference
procedure requires no supermodularity, no equilibrium selection, and
no computation of equilibrium networks, for simulations we do need
to generate the simulated equilibrium networks that our inference
procedure takes as data. Exploiting the supermodular structure enables
us to compute large equilibrium networks even with $n=10{,}000$.

Throughout this section, we compute various CSs as described in Section
\ref{sec:single_inf} and evaluate their performance across the $K=500$
repetitions, in every case at the 95\% nominal level. In the parametric case, the standard logistic CDF fixes the location and scale. We work with a discrete grid on the parameter space:
$b_{0}-b_{0,\mathrm{true}}\in[-7,7]$ with step size $0.25$, $\beta\in[-3.25,1.25]$
with step size $0.125$, and $\gamma\in[-16,72]$ with step size $0.125$. The test statistic is the Berk--Jones statistic in Section~\ref{subsec:gendom},
with the thresholds $c$ fixed on $731$ points spanning $[-36.5,36.5]$
in steps of $0.1$, intersected with the constrained
threshold sets $C(z_{D};\theta,Z)$ at each candidate. The conditional
critical value is simulated accordingly with $B=999$, and every nonempty
type-pair cell is retained. In the semiparametric
case, we normalize the
scale of the homophily effect parameter to unity, $|\beta|=1$,
fix the location by setting the intercept term $b_{0}$ to zero, and
retain the cells with at least $\underline{m}=20$ dyads.

\subsection{Parametric Confidence Sets}

\label{subsec:sim_master}

Table~\ref{tab:sim_master} reports the results for the parametric
CSs, with the reported diagnostics defined in the notes to the table.

\begin{table}[t]
\centering
\setlength{\tabcolsep}{4pt}
 \caption{Parametric confidence sets}
\label{tab:sim_master} \begin{threeparttable}%
\begin{tabular}{@{}lccccccc@{}}
\toprule
\multicolumn{1}{@{}l}{$n$} & \multicolumn{1}{c}{$\gamma$ sign} & \multicolumn{1}{c}{$\gamma$ width} & \multicolumn{1}{c}{vac} & \multicolumn{1}{c}{$\gamma$ edge L\,/\,U} & \multicolumn{1}{c}{$\beta$ sign} & \multicolumn{1}{c}{$\beta$ width} & \multicolumn{1}{c@{}}{$b_{0}$ width}\tabularnewline
\midrule
$100$ & $0.748$ & $48.31$ & $0.074$ & $0.106$ / $0.176$ & $0.182$ & $2.25$ & $7.00$\tabularnewline
$200$ & $0.948$ & $33.44$ & $0.012$ & $0.026$ / $0.024$ & $0.818$ & $1.50$ & $5.00$\tabularnewline
$400$ & $1.000$ & $24.88$ & $0.000$ & $0.000$ / $0.002$ & $0.990$ & $0.94$ & $3.50$\tabularnewline
$800$ & $1.000$ & $18.13$ & $0.000$ & $0.000$ / $0.000$ & $1.000$ & $0.63$ & $2.50$\tabularnewline
$1{,}600$ & $1.000$ & $12.63$ & $0.000$ & $0.000$ / $0.000$ & $1.000$ & $0.25$ & $1.50$\tabularnewline
$3{,}200$ & $1.000$ & $7.75$ & $0.000$ & $0.000$ / $0.000$ & $1.000$ & $0.13$ & $1.00$\tabularnewline
$6{,}400$ & $1.000$ & $2.88$ & $0.000$ & $0.000$ / $0.000$ & $1.000$ & $0.00$ & $0.00$\tabularnewline
$10{,}000$ & $1.000$ & $2.25$ & $0.000$ & $0.000$ / $0.000$ & $1.000$ & $0.00$ & $0.00$\tabularnewline
\bottomrule
\end{tabular}\begin{tablenotes}[flushleft]
\raggedright

{\footnotesize\item \emph{Notes:} ``Sign'' is the share of
replications whose projection lies strictly on the correct side of zero;
``width'' the median total length of the accepted projection, which
need not be an interval; ``vac'' the share accepting the whole
$\gamma$ grid; ``edge L\,/\,U'' the shares touching its lower or upper endpoint.} \end{tablenotes}
\end{threeparttable}
\end{table}

Our inference procedure is demonstrably
informative in a single-network setting even at the smallest size
$n=100$. The true parameter vector is covered in every one of the
$K=500$ replications at each tested $n$, so we omit a coverage column from the table. The CSs exhibit conservative empirical coverage in these simulations, consistent with our theoretical guarantee of at least nominal coverage in finite samples. At the same time, the CSs
are nontrivial subsets of the grid, contract with $n$, and deliver
sign-revealing information. The median $\gamma$-projection width
falls at every tested $n$, while the vacuous-set rate falls from $7.4\%$
at $n=100$ to zero from $n=400$ onward.

Notably, the sign of the strategic coefficient $\gamma$, an important
object of interest in the strategic network formation literature
\citep[e.g.,][]{de2012inference,sheng2020structural},
is certified in $74.8\%$ of replications at $n=100$, in $94.8\%$
at $n=200$, and in every replication from $n=400$ onward, while
the median $\gamma$ projection width falls from $48.31$ at $n=100$ to
$2.25$ at $n=10{,}000$, on a candidate grid of length $88$.
To the best of our knowledge, single-network finite-sample valid and sign-revealing CSs were not previously available in this environment.

The CSs also deliver increasingly tight bounds on the homophily coefficient $\beta$ and the intercept $b_{0}$. The negative sign of $\beta$ is certified in $18.2\%$ of replications at $n=100$, $81.8\%$ at $n=200$, $99.0\%$ at $n=400$, and in every replication from $n=800$ onward. The median width falls from $2.25$ for $\beta$ and $7.00$ for $b_{0}$, half of the respective candidate grids, to zero at $n=6{,}400$ for both parameters. Thus the nuisance-parameter projections concentrate considerably faster than the projection for the strategic coefficient.

\subsection{Semiparametric Confidence Sets}

\label{subsubsec:sim_ladder}

In Table \ref{tab:sim_ladder}, we report results on the semiparametric
versions of our inference procedure. The semiparametric column is
calibrated by conditional simulation in the cellwise form introduced in
Footnote~\ref{fn:kappahat} and explained further in Supplemental Appendix S.1. We also include corresponding results about the parametric case from the last subsection for ease of comparison.

These CSs are computed on the same simulated networks and use the same
$\gamma\in[-16,72]$ grid with step size $0.125$. The parametric
procedure treats $(b_{0},\beta,\gamma)$ as unknown under the standard-logistic
location and scale normalization. The fully semiparametric procedure
instead fixes $b_{0}=0$ as a location normalization and sets $|\beta|=1$
as a scale normalization. For each sign of $\beta$, we use a threshold grid with step size $0.0625$ and a range chosen to cover all candidate index values. Width is the median total length of the
accepted $\gamma$ projection in both columns. Vacuity refers to the
$\gamma$ projection in both columns.

\begin{table}[t]
\centering
\setlength{\tabcolsep}{4pt}
 \caption{Semiparametric confidence sets: $\gamma$ projections}
\label{tab:sim_ladder}
\begin{threeparttable}%
\begin{tabular}{@{}lcccccc@{}}
\toprule
\multicolumn{1}{@{}l}{$n$} & \multicolumn{3}{c}{parametric} & \multicolumn{3}{c@{}}{semiparametric}\tabularnewline
\cmidrule(lr){2-4}\cmidrule(lr){5-7}
\multicolumn{1}{@{}l}{} & \multicolumn{1}{c}{sign} & \multicolumn{1}{c}{width} & \multicolumn{1}{c}{vac} & \multicolumn{1}{c}{sign} & \multicolumn{1}{c}{width} & \multicolumn{1}{c@{}}{vac}\tabularnewline
\midrule
$100$ & $0.748$ & $48.31$ & $0.074$ & $0.416$ & $72.00$ & $0.094$\tabularnewline
$200$ & $0.948$ & $33.44$ & $0.012$ & $0.560$ & $67.63$ & $0.018$\tabularnewline
$400$ & $1.000$ & $24.88$ & $0.000$ & $0.666$ & $53.94$ & $0.000$\tabularnewline
$800$ & $1.000$ & $18.13$ & $0.000$ & $0.716$ & $40.00$ & $0.000$\tabularnewline
$1{,}600$ & $1.000$ & $12.63$ & $0.000$ & $0.868$ & $23.38$ & $0.000$\tabularnewline
$3{,}200$ & $1.000$ & $7.75$ & $0.000$ & $0.978$ & $20.00$ & $0.000$\tabularnewline
$6{,}400$ & $1.000$ & $2.88$ & $0.000$ & $0.998$ & $16.25$ & $0.000$\tabularnewline
$10{,}000$ & $1.000$ & $2.25$ & $0.000$ & $0.998$ & $12.50$ & $0.000$\tabularnewline
\bottomrule
\end{tabular}
\end{threeparttable}
\end{table}

The semiparametric CSs, while less tight than their parametric counterparts as anticipated, remain nontrivial and informative. The true $\gamma_0$ is covered in every replication, so the coverage column is similarly omitted as in Table \ref{tab:sim_master}. The cleanest comparison between the two CSs concerns the \emph{sign rate} for $\gamma_0$, which orders the
CSs by the strength of the imposed assumptions. Despite leaving \(F_{\varepsilon}\) unrestricted, the semiparametric CSs increasingly reveal the positive sign of \(\gamma_0\): the sign rate rises from \(41.6\%\) at \(n=100\) to \(97.8\%\) at \(n=3{,}200\), reaching \(99.8\%\) (\(499\) of \(500\) replications) at both \(n=6{,}400\) and \(n=10{,}000\). By comparison, the parametric CSs certify the sign in every replication from \(n=400\) onward. The semiparametric projections remain wider but contract substantially, with median width falling from \(72.00\) at \(n=100\) to \(12.50\) at \(n=10{,}000\). Although boundary contact may understate their widths at small \(n\), the projections contract within the candidate range at larger sizes. The vacuity rate falls from \(9.4\%\) at \(n=100\) to \(1.8\%\) at \(n=200\), with no vacuous projections from \(n=400\) onward.

The Supplemental Appendix examines shock dependence, multidimensional exogenous covariates, and symmetry restrictions on \(F_{\varepsilon}\). Additive common shocks leave performance largely unchanged, whereas block dependence outside the maintained assumptions leads to undercoverage. Additional covariates preserve informative sign results at larger network sizes, while imposing symmetry improves sign revelation relative to the unrestricted semiparametric procedure.

\section{Empirical Applications}

\label{sec:application}

We apply the parametric CS to two empirical data sets: contact
(and friendship) among high-school students, and friendship among
Twitch streamers. In each application we ask whether link formation
responds to the links of others, by constructing a finite-sample valid
CS for the coefficient on a common-friend statistic while leaving
equilibrium selection and network dependence unrestricted. In both
data sets that coefficient is bounded away from zero from below at
the 95\% confidence level, so triad closure is an economically operative
force in these networks and not an artifact of homophily on the observed
covariates. This distinction matters for policy: interventions that
seed or remove a small number of links have knock-on effects on the
remaining ones precisely when the strategic channel is present, whereas
under pure homophily their effects would remain local. In both applications,
each reported row is the projection of a single joint 95\% confidence
set over all coefficients.

\subsection{High-School Contact and Friendship Networks}

\label{subsec:app_hs}

The Thiers13 data contain $327$ students, and hence $53{,}301$ dyads,
in nine classes and four class types \citep{mastrandrea2015contact}.\footnote{The data set is publicly available and can be accessed at \href{https://sociopatterns.org/datasets/high-school-contact-and-friendship-networks/}{https://sociopatterns.org/datasets/high-school-contact-and-friendship-networks/}.}
Contact sensors record face-to-face proximity in twenty-second intervals
during one school week, and we set $Y_{ij}=1$ when a pair has at
least $\underline{\kappa}$ recorded intervals in total. The baseline
is $\underline{\kappa}=2$, with $\underline{\kappa}=1$ and $\underline{\kappa}=3$
also reported as robustness checks. The friendship layer sets $Y_{ij}=1$
if either student nominates the other in a survey.\footnote{\label{fn:hs_response}Only
$135$ of the $327$ students returned the friendship survey, so the
friendship layer is a partially observed version of the underlying
nomination network. We therefore treat the contact networks as the
primary specification and report the friendship results as supporting
evidence.} The exogenous covariates
are an indicator for different classes ($D_{ij}$) and an indicator
for different class types ($T_{ij}$), so $Z_{ij}=(1,D_{ij},T_{ij})$, and the endogenous covariate
is the common-friend count in the same network, normalized as in the
simulations and reported below in per-common-friend units. Selected summary
statistics are reported in Table \ref{tab:hs_summary}.

\begin{table}[t]
\centering
\setlength{\tabcolsep}{4pt}
 \caption{High-school network: summary statistics}
\label{tab:hs_summary} \begin{threeparttable}%
\begin{tabular}{@{}lccccc@{}}
\toprule
network & \multicolumn{1}{c}{links} & \multicolumn{1}{c}{density (\%)} & \multicolumn{1}{c}{mean degree} & \multicolumn{1}{c}{CF mean (sd)} & \multicolumn{1}{c@{}}{$\mathrm{CF}\geq1$ (\%)}\tabularnewline
\midrule 
contact, $\underline{\kappa}=1$  & $5{,}818$  & $10.92$  & $35.58$  & $4.33$ $(6.64)$  & $73.20$\tabularnewline
contact, $\underline{\kappa}=2$  & $4{,}031$  & $7.56$  & $24.65$  & $2.11$ $(4.62)$  & $40.91$\tabularnewline
contact, $\underline{\kappa}=3$  & $3{,}337$  & $6.26$  & $20.41$  & $1.46$ $(3.59)$  & $32.11$\tabularnewline
friendship  & $406$  & $0.76$  & $2.48$  & $0.05$ $(0.43)$  & $2.47$\tabularnewline
\bottomrule
\end{tabular} \end{threeparttable}
\end{table}

Again, we use the Berk--Jones statistic with constrained thresholding
in Section~\ref{subsec:gendom}. The conditioning cells are the 45
unordered pairs of the nine classes, with a cell floor of 20 dyads. The critical
values use $B=999$ conditional simulation draws. The endogenous statistic is
$X=\mathrm{CF}/(n-2)$. We use the normalized common-friend count $X=\mathrm{CF}/(n-2)$, with grid spacing $1.625$ for its coefficient $\gamma$. For ease of interpretation, we report $\gamma/(n-2)$, the increase in the link index per additional common friend. Since $n-2=325$, the reported grid spacing is $0.005$. We search over a rectangular parameter grid, with spacing $0.25$ for the remaining coefficients. All reported confidence-set projections lie strictly inside their respective search ranges, so the boundaries are nonbinding. The threshold grid has spacing $0.125$ and covers the full range of the link index in each specification.

\begin{table}[t]
\centering
\setlength{\tabcolsep}{4pt}
 \caption{High-school network: parametric confidence-set projections}
\label{tab:hs_cs} \begin{threeparttable}%
\begin{tabular}{@{}lccc@{}}
\toprule
network & \multicolumn{1}{c}{$b_{D}$} & \multicolumn{1}{c}{$b_{T}$} & \multicolumn{1}{c@{}}{$\gamma$}\tabularnewline
\midrule 
contact, $\underline{\kappa}=1$ & $[-4.250,2.750]$ & $[-4.000,2.500]$ & $[0.020,0.335]$\tabularnewline
contact, $\underline{\kappa}=2$ & $[-4.500,1.250]$ & $[-3.500,1.000]$ & $[0.045,0.370]$\tabularnewline
contact, $\underline{\kappa}=3$ & $[-4.750,0.750]$ & $[-3.500,0.500]$ & $[0.025,0.405]$\tabularnewline
friendship & $[-4.500,2.750]$ & $[-5.250,3.250]$ & $[0.110,2.270]$\tabularnewline
\bottomrule
\end{tabular} \end{threeparttable}
\end{table}

Table \ref{tab:hs_cs} reports the parametric 95\%-level confidence-set
projection results under the assumption that $\varepsilon_{ij}$ is
standard logistic. 
All four common-friend projections are positive. Equivalently, the
$\gamma=0$ slice of each reported joint CS, which is the
CS for the corresponding $Z$-only logistic specification
under the same inversion, is empty, so by Remark \ref{rem:over_id}
the specification without the common-friend term is rejected at the
5\% level at every nuisance-coefficient value on the reported grid.
A positive $\gamma$ indicates the tendency for triad closure in strategic
network formation: sharing a contact (or friend) raises the net surplus
from linking. In the baseline $\underline{\kappa}=2$ contact network
one additional common friend raises the structural link index by $0.045$
to $0.370$. At the sample mean of $2.11$ common friends, the
endogenous term contributes roughly $0.09$ to $0.78$ index units,
relative to the standard logistic error scale. The $\underline{\kappa}=1$
and $\underline{\kappa}=3$ rows are also strictly positive, so the
sign of $\gamma$ does not depend on a single contact threshold. The friendship projection for $\gamma_0$ is also positive but much wider than those for the contact networks. Partial survey response (Footnote~\ref{fn:hs_response}) limits identifying variation: $97.53\%$ of dyads have no common friend, and $193$ of $327$ students are isolated.

\subsection{Twitch Friendship Networks}

\label{subsec:app_twitch}

The Twitch data set contains information about networks of gamers
who stream in a certain language (DE, EN, ES, FR, PT, and RU): six
undirected friendship networks, one per streamer language \citep{rozemberczki2021multiscale}.\footnote{The data set is publicly available and can be accessed at \href{https://snap.stanford.edu/data/twitch-social-networks.html}{https://snap.stanford.edu/data/twitch-social-networks.html}.}
We  treat each network as a separate inference problem with language-specific
parameters, and the CS is computed separately for each network. See
Table \ref{tab:twitch_summary} for some key summary statistics.

\begin{table}[t]
\centering
\setlength{\tabcolsep}{4pt}
 \caption{Twitch gamer network: summary statistics}
\label{tab:twitch_summary} \begin{threeparttable}%
\begin{tabular}{@{}lccccc@{}}
\toprule
language & \multicolumn{1}{c}{nodes} & \multicolumn{1}{c}{links} & \multicolumn{1}{c}{density (\%)} & \multicolumn{1}{c}{CF mean (sd)} & \multicolumn{1}{c@{}}{$\mathrm{CF}\geq1$ (\%)}\tabularnewline
\midrule 
German (DE)  & $9{,}498$  & $153{,}138$  & $0.34$  & $0.86$ $(2.51)$  & $37.21$\tabularnewline
English (EN)  & $7{,}126$  & $35{,}324$  & $0.14$  & $0.08$ $(0.41)$  & $6.31$\tabularnewline
Spanish (ES)  & $4{,}648$  & $59{,}382$  & $0.55$  & $0.66$ $(2.14)$  & $27.72$\tabularnewline
French (FR)  & $6{,}549$  & $112{,}666$  & $0.53$  & $1.09$ $(2.95)$  & $38.94$\tabularnewline
Portuguese (PT)  & $1{,}912$  & $31{,}299$  & $1.71$  & $2.18$ $(5.06)$  & $51.32$\tabularnewline
Russian (RU)  & $4{,}385$  & $37{,}304$  & $0.39$  & $0.46$ $(1.52)$  & $23.25$\tabularnewline
\bottomrule
\end{tabular}
\end{threeparttable}
\end{table}

We construct the exogenous covariates as follows. Let $q^{V}_{i}\in\{1,\ldots,5\}$
be streamer $i$'s within-language quintile of total views and $m_{i}$
the mature-content indicator. Define
\[
V_{ij}=|q^{V}_{i}-q^{V}_{j}|,\qquad S_{ij}=q^{V}_{i}+q^{V}_{j}-6,\qquad M_{ij}=m_{i}+m_{j}-1,
\]
so that $V_{ij}\in\{0,\ldots,4\}$ measures the popularity gap, $S_{ij}\in\{-4,\ldots,4\}$
captures the popularity level of the pair, and $M_{ij}\in\{-1,0,1\}$
the centered mature-content pair count. All three are integer-valued
and symmetric in $(i,j)$, and together they give the $15$ unordered
views-quintile pairs crossed with the three mature-pair values. The structural index in the language-$\ell$ network is
then
\[
\delta_{ij,\ell}=b_{0\ell}+b_{V\ell}V_{ij}+b_{S\ell}S_{ij}+b_{M\ell}M_{ij}+\gamma_{\ell}\,X_{ij,\ell},
\]
with $Z_{ij}=(1,V_{ij},S_{ij},M_{ij})$ and the endogenous regressor  $X_{ij}$ given by
\[
X_{ij,\ell}:=\min\big\{\mathrm{CF}_{ij},\,q^{\mathrm{CF}}_{95,\ell}\big\},
\]
where $q^{\mathrm{CF}}_{95,\ell}$ is the 95th percentile of the raw
common-friend counts in network $\ell$ (the $q_{95}$ column of Table
\ref{tab:twitch_tail}): $q^{\mathrm{CF}}_{95,\ell}=4,1,3,5,9,$ and
$2$ for DE, EN, ES, FR, PT, and RU, respectively. Since the cap is computed from the
realized network, $X_{ij,\ell}$ is itself a globally defined endogenous
statistic, which our framework explicitly accommodates as discussed
in Footnote~\ref{fn:own_link}.

We cap the common-friend count at the network's own tail quantile
to deal with the highly skewed and heterogeneous distribution of $\mathrm{CF}_{ij}$
within each language network and across the six language networks.
Table~\ref{tab:twitch_tail} documents how extreme that skewness
is. In the German network, for instance, $62.79\%$ of dyads have no
common friend at all and the 99th percentile is $9$, yet the maximum
$\mathrm{CF}_{ij}$ is $1{,}432$. An uncapped count would therefore let extreme neighborhoods dominate identifying variation. The transformation $\log(1+\mathrm{CF}_{ij})$ accommodates zeros but imposes diminishing returns. Capping limits the upper tail while preserving zeros and constant per-friend index increments below the cap. All six networks
are heavily skewed, but the tails differ substantially across the
six: the 95th percentile ranges from $1$ (English) to $9$ (Portuguese),
so a common cap would represent very different events across networks.
Capping at the network's own 95th percentile instead preserves the
dyad-level variation that identifies $\gamma_{\ell}$ while adapting
to each network's tail: under this specification $\gamma_{\ell}$
is the change in the latent link index per additional common friend
below the cap. 

We use fine threshold and parameter grids with steps of $0.125$ and $0.025$, respectively. We use sufficiently wide parameter ranges so that no reported projection reaches a  grid boundary. The
conditioning cells are the 45 combinations of the 15  
views-quintile pairs and the three values of $M_{ij}$, subject to a
20-dyad cell floor which no cell reaches. The critical values use $B=999$ conditional simulation draws.

\begin{table}[t]
\centering
\setlength{\tabcolsep}{4pt}
\caption{Twitch gamer network: right tails of $\mathrm{CF}$ distribution }
\label{tab:twitch_tail}\begin{threeparttable}%
\begin{tabular}{@{}lcccccc@{}}
\toprule
language & \multicolumn{1}{c}{$\mathrm{CF}=0$ (\%)} & \multicolumn{1}{c}{$q_{90}$} & \multicolumn{1}{c}{$q_{95}$} & \multicolumn{1}{c}{$q_{99}$} & \multicolumn{1}{c}{$q_{99.9}$} & \multicolumn{1}{c@{}}{$\max\mathrm{CF}$}\tabularnewline
\midrule 
German (DE) & $62.79$ & $2$ & $4$ & $9$ & $24$ & $1{,}432$\tabularnewline
English (EN) & $93.69$ & $0$ & $1$ & $2$ & $4$ & $134$\tabularnewline
Spanish (ES) & $72.28$ & $2$ & $3$ & $8$ & $24$ & $436$\tabularnewline
French (FR) & $61.06$ & $3$ & $5$ & $11$ & $26$ & $1{,}095$\tabularnewline
Portuguese (PT) & $48.68$ & $6$ & $9$ & $20$ & $54$ & $449$\tabularnewline
Russian (RU) & $76.75$ & $1$ & $2$ & $6$ & $14$ & $562$\tabularnewline
\bottomrule
\end{tabular}\end{threeparttable}
\end{table}

\begin{table}[t]
\centering
\setlength{\tabcolsep}{3pt}
 \caption{Twitch gamer network: parametric confidence-set projections}
\label{tab:twitch_cs} \begin{threeparttable}%
\begin{tabular}{@{}lcccc@{}}
\toprule
language & \multicolumn{1}{c}{$b_{V}$ (views gap)} & \multicolumn{1}{c}{$b_{S}$ (views level)} & \multicolumn{1}{c}{$b_{M}$ (mature pair)} & \multicolumn{1}{c@{}}{$\gamma$ (capped CF)}\tabularnewline
\midrule 
German (DE) & $[0.025,0.600]$ & $[0.175,0.725]$ & $[-0.650,0.425]$ & $[0.350,2.250]$\tabularnewline
English (EN) & $[0.100,0.425]$ & $[0.350,0.600]$ & $[-0.275,0.225]$ & $[0.650,8.475]$\tabularnewline
Spanish (ES) & $[-0.050,0.725]$ & $[0.125,0.750]$ & $[-1.175,0.800]$ & $[0.325,2.850]$\tabularnewline
French (FR) & $[-0.275,0.800]$ & $[0.050,0.875]$ & $[-0.900,0.900]$ & $[0.225,1.825]$\tabularnewline
Portuguese (PT) & $[-0.250,1.125]$ & $[0.000,1.100]$ & $[-1.225,1.175]$ & $[0.100,1.000]$\tabularnewline
Russian (RU) & $[0.175,0.750]$ & $[0.275,0.800]$ & $[-0.600,0.600]$ & $[0.400,4.100]$\tabularnewline
\bottomrule
\end{tabular} \end{threeparttable}
\end{table}

Table \ref{tab:twitch_cs} reports the parametric confidence-set projections
under the assumption that $\varepsilon_{ij}$ is standard logistic.
Every language produces a strictly positive CS for $\gamma_{\ell}$: under the maintained capped-index
specification, each additional common friend below the cap raises
the latent link index. The
exclusion of zero also has a restricted-model interpretation. When
$\gamma_{\ell}=0$, $X_{ij,\ell}$ drops out of the index, so the zero
slice of the joint inversion is exactly the CS for the
$Z$-only logistic specification, computed with the same cells, threshold
grid, and critical value. This slice is empty in all six language
networks: by the emptiness argument of Remark \ref{rem:over_id},
the $Z$-only specification is rejected in each network over the stated coefficient ranges at the 5\% level. These results highlight the importance of accounting for endogenous network variables, such as common-friend counts, in this application.

The per-friend lower endpoints range from $0.100$ (Portuguese) to $0.650$ (English). Under the standard-logistic normalization, the German projection $[0.350,2.250]$, for instance, implies a latent-index increase of $0.350$ to $2.250$ per additional common friend below the cap, while the English lower bound is an increase of $0.650$ from the first common friend. Because common friends are endogenous, these are not conditional log-odds or odds-ratio effects. The exogenous-coefficient projections also provide sign information: the views-level projections are strictly positive in five networks, and the views-gap projections are strictly positive in the German, English, and Russian networks. The mature-content projections include zero in all six networks. Across networks with very different sizes, densities, and tail behaviors, the procedure delivers informative bounds on the strategic coefficient and reveals the signs of several effects associated with observed streamer characteristics.

\section{Conclusion}\label{sec:con}

This paper develops single-network and finite-sample valid inference for structural parameters
in strategic network formation models.
The bounding-by-$c$ device converts restrictions involving endogenous
network statistics into realization-wise inequalities whose middle
term depends only on exogenous covariates and i.i.d.\ pairwise shocks.
This yields simulable critical values and finite-sample CSs
in both parametric and semiparametric settings, without the need
to solve, simulate, or enumerate equilibrium networks. The validity
of the procedure does not require restrictions on network density, network dependence, or
equilibrium selection.
Its computational simplicity also makes the procedure scalable to
large networks.

The core idea of bounding by $c$ might be useful in other contexts as well, especially in settings with complicated and intractable endogeneity. Potential applications include nontransferable-utility network models with bilateral consent, as well as matching, auctions, and entry games. Developing these model-specific bounds and assessing their informativeness are left for future research.

\section*{Appendix: Proofs}
\begin{proof}[\textbf{Proof of Lemma~\ref{lem:exo_middle}}]
Fix a dyad type $z_{D}$, write $N_{n}:=\binom{n}{2}$ for the number
of dyads, and define, for $c\in\R$, 
\begin{align*}
\hat{A}_{n,z_{D}}(c) & :=\frac{1}{N_{n}}\sum_{ij}\ind\{\varepsilon_{ij}\leq c\}\ind\{Z_{D,ij}=z_{D}\},\quad\\
\hat{A}^{-}_{n,z_{D}}(c) & :=\frac{1}{N_{n}}\sum_{ij}\ind\{\varepsilon_{ij}<c\}\ind\{Z_{D,ij}=z_{D}\},\\
\hat{D}_{n,z_{D}} & :=N^{-1}_{n}\sum_{ij}\ind\{Z_{D,ij}=z_{D}\}\equiv N^{-1}_{n}\hat{m}\left(z_{D}\right),
\end{align*}
and $\hat{F}_{n}(c;z_{D})=\hat{A}_{n,z_{D}}(c)/\hat{D}_{n,z_{D}},$
$\hat{F}^{-}_{n}(c;z_{D}):=\hat{A}^{-}_{n,z_{D}}(c)/\hat{D}_{n,z_{D}}$
whenever $\hat{D}_{n,z_{D}}>0$. 

We first establish pointwise convergence results at each $z_{D}$
and each $c$. Each of $\hat{A}_{n,z_{D}}(c)$, $\hat{A}^{-}_{n,z_{D}}(c)$,
and $\hat{D}_{n,z_{D}}$ is an average over the $N_{n}$ dyads of
a bounded kernel of the array $W_{ij}:=(Z_{i},Z_{j},\varepsilon_{ij})$.
Under Assumptions~\ref{ass:iid_errors}--\ref{ass:random_sampling}
each of the three scalar summand arrays (for fixed $c$ and $z_{D}$)
is symmetric in $(i,j)$ by the symmetry of the partition cells, jointly
exchangeable (i.e., its distribution is invariant under relabeling
of the agents), and dissociated (entries indexed by disjoint agent
sets are independent). The strong law of large numbers for $U$-statistics
of such arrays \citep[e.g., Theorem 3 of][]{eagleson1978limit} therefore
gives, almost surely, 
\[
\hat{A}_{n,z_{D}}(c)\to F_{\varepsilon}(c)\,q(z_{D}),\quad\hat{A}^{-}_{n,z_{D}}(c)\to F_{\varepsilon}(c^{-})\,q(z_{D}),\quad\hat{D}_{n,z_{D}}\to q(z_{D})>0,
\]
where $q(z_{D}):=\P(Z_{D,ij}=z_{D})$, $F_{\varepsilon}(c^{-}):=\lim_{u\uparrow c}F_{\varepsilon}(u)$
and the products on the right-hand side, e.g., $\E[\ind\{\varepsilon_{ij}\leq c\}\ind\{Z_{D,ij}=z_{D}\}]=F_{\varepsilon}(c)q(z_{D})$,
follow from the independence between $\varepsilon$ and $Z$ (Assumption~\ref{ass:exogeneity}).
Hence, $\hat{F}_{n}(c;z_{D})\to F_{\varepsilon}(c)$ and $\hat{F}^{-}_{n}(c;z_{D})\to F_{\varepsilon}(c^{-})\text{ }$for
each fixed $c$, almost surely.

Next, we establish uniformity in $c$ at each $z_{D}$. Fix an integer
$M\geq2$ and let $c_{j}:=\inf\{c:F_{\varepsilon}(c)\geq j/M\}$,
$j=1,\ldots,M-1$, denote the $M$-quantile grid of $F_{\varepsilon}$;
each $c_{j}$ is finite because $F_{\varepsilon}$ tends to $0$ and
$1$ at $\mp\infty$. Let $\Omega_{M,z_{D}}$ be the probability-one
event on which $\hat{D}_{n,z_{D}}>0$ eventually and the pointwise
convergence established above holds for both $\hat{F}_{n}$ and $\hat{F}^{-}_{n}$
at every $c_{j}$. Define
\[
\Delta_{n,z_{D},M}:=\max_{1\leq j\leq M-1}\max\left\{ \left|\hat{F}_{n}(c_{j};z_{D})-F_{\varepsilon}(c_{j})\right|,\ \left|\hat{F}^{-}_{n}(c_{j};z_{D})-F_{\varepsilon}(c^{-}_{j})\right|\right\} .
\]
Since the grid is finite, $\Delta_{n,z_{D},M}\to0$ on $\Omega_{M,z_{D}}$.
To extend this control from the grid to every $c\in\R$, fix $c$.
For the lower bound on $\hat{F}_{n}(c;z_{D})$, if $F_{\varepsilon}(c)<1/M$
then $F_{\varepsilon}(c)-\hat{F}_{n}(c;z_{D})<1/M$ immediately. Otherwise,
set $j=\min\{\lfloor MF_{\varepsilon}(c)\rfloor,M-1\}$. Then $c_{j}\leq c$
and $F_{\varepsilon}(c_{j})\geq j/M\geq F_{\varepsilon}(c)-1/M$,
so 
\[
\hat{F}_{n}(c;z_{D})\geq\hat{F}_{n}(c_{j};z_{D})\geq F_{\varepsilon}(c)-\frac{1}{M}-\Delta_{n,z_{D},M}.
\]
For the upper bound, if $F_{\varepsilon}(c)\geq(M-1)/M$ then $\hat{F}_{n}(c;z_{D})-F_{\varepsilon}(c)\leq1/M$.
Otherwise, set $j=\lfloor MF_{\varepsilon}(c)\rfloor+1\leq M-1$.
Because $j/M>F_{\varepsilon}(c)$, the definition of $c_{j}$ implies
$c<c_{j}$ and $F_{\varepsilon}(c^{-}_{j})\leq j/M$. Hence, we have
\[
\hat{F}_{n}(c;z_{D})\leq\hat{F}^{-}_{n}(c_{j};z_{D})\leq F_{\varepsilon}(c)+\frac{1}{M}+\Delta_{n,z_{D},M}.
\]
Combining the two bounds gives $\sup_{c\in\R}\big|\hat{F}_{n}(c;z_{D})-F_{\varepsilon}(c)\big|\ \leq\ \Delta_{n,z_{D},M}+\tfrac{1}{M}$
$\text{on }\Omega_{M,z_{D}}.$

Finally, ${\cal Z}_{D}$ is finite with $q(z_{D})>0$ for any $z_{D}$
by Definition \ref{def:disc}, so the intersection $\overline{\Omega}_{M}$
of the finitely many $\Omega_{M,z_{D}}$ still has probability one,
and the bound above holds on it with $\Delta_{n,M}:=\max_{z_{D}\in{\cal Z}_{D}}\Delta_{n,z_{D},M}\to0$.
Since $M$ is arbitrary and $\bigcap_{M\geq2}\overline{\Omega}_{M}$
is again a probability-one event, $\max_{z_{D}\in\mathcal{Z}_{D}}\sup_{c\in\R}\big|\hat{F}_{n}(c;z_{D})-F_{\varepsilon}(c)\big|\to0$
as $n\to\infty.$
\end{proof}

\begin{proof}[\textbf{Proof of Theorem~\ref{thm:single_pop}}]
Taking $\limsup_{n}$ on the left of \eqref{eq:empirical_sandwich}
and applying Lemma~\ref{lem:exo_middle}, 
\[
p^{\infty}_{L}(z_{D},c;\theta_{0})=\limsup_{n}\ \hat{p}^{\,(n)}_{L}(z_{D},c;\theta_{0})\ \leq\ \limsup_{n}\ \hat{F}_{n}(c;z_{D})\ =\ F_{\varepsilon}(c),
\]
and symmetrically $p^{\infty}_{U}(z_{D},c;\theta_{0})=\liminf_{n}\hat{p}^{\,(n)}_{U}(z_{D},c;\theta_{0})\geq\liminf_{n}\hat{F}_{n}(c;z_{D})=F_{\varepsilon}(c)$,
which together imply that $p^{\infty}_{L}(z_{D},c;\theta_{0})\leq F_{\varepsilon}(c)\leq p^{\infty}_{U}(z_{D},c;\theta_{0})$
for all $c\in\R$ and for all $z_{D}$. Then $\max_{z_{D}}p^{\infty}_{L}(z_{D},c;\theta_{0})\leq F_{\varepsilon}(c)\leq\min_{z_{D}}p^{\infty}_{U}(z_{D},c;\theta_{0}),$
giving part (i). Part (ii) follows by inserting $F_{\varepsilon}(\cdot\,;\theta_{\varepsilon,0})=F_{\varepsilon}$
as the middle term. 
\end{proof}

\begin{proof}[\textbf{Proof of Lemma~\ref{lem:dom_semi}}]
Fix any $c\in\R$. The finite-sample sandwich \eqref{eq:empirical_sandwich}
holds cell by cell at $\theta_{0}$: $\hat{p}^{\,(n)}_{L}(z_{D},c;\theta_{0})\,\leq\,\hat{F}_{n}(c;z_{D})\,\leq\,\hat{p}^{\,(n)}_{U}(z_{D},c;\theta_{0})$
for every $z_{D}\in\hat{{\cal Z}}_{D}$. Maximizing the left inequality
and minimizing the right one over $z_{D}\in\hat{{\cal Z}}_{D}$ yields
\[
\max_{z_{D}}\hat{p}^{\,(n)}_{L}(z_{D},c;\theta_{0})\ \leq\ \max_{z_{D}}\hat{F}_{n}(c;z_{D}),\quad\min_{z_{D}}\hat{p}^{\,(n)}_{U}(z_{D},c;\theta_{0})\ \geq\ \min_{z_{D}}\hat{F}_{n}(c;z_{D}),
\]
and hence 
\[
\max_{z_{D}}\hat{p}^{\,(n)}_{L}(z_{D},c;\theta_{0})-\min_{z_{D}}\hat{p}^{\,(n)}_{U}(z_{D},c;\theta_{0})\ \leq\ \max_{z_{D}}\hat{F}_{n}(c;z_{D})-\min_{z_{D}}\hat{F}_{n}(c;z_{D}).
\]
The right-hand side is at most $R_{n}$ by the definition \eqref{eq:Rn_range}
and is nonnegative. Hence, taking the supremum over $c$ and applying
the positive part function preserves the bound, implying $\widehat{Q}_{n}(\theta_{0})\leq R_{n}$. 
\end{proof}

\begin{proof}[\textbf{Proof of Theorem~\ref{thm:semi_CS}}]
We condition on a fixed realization of $Z=\left(Z_{i}\right)^{n}_{i=1}$,
and suppress the qualifier ``almost surely'' subsequently. Then,
$\hat{{\cal Z}}_{D}$ and the cell sizes $\hat{m}(z_{D})$ are all
fixed. Let $(\nu_{ij})_{ij}$ be i.i.d. $\mathrm{Unif}[0,1]$ random
variables, independent of everything else, and define the \emph{distributional
transform} 
\[
U_{ij}:=F_{\varepsilon}(\varepsilon^{-}_{ij})+\nu_{ij}\big[F_{\varepsilon}(\varepsilon_{ij})-F_{\varepsilon}(\varepsilon^{-}_{ij})\big].
\]
When $F_{\varepsilon}$ is continuous, the above specializes to the
standard probability integral transform $U_{ij}:=F_{\varepsilon}(\varepsilon_{ij})$.
Regardless of whether $F_{\varepsilon}$ is continuous or not, it
is a standard result that $U_{ij}\sim\mathrm{Unif}[0,1]$ and $\varepsilon_{ij}=F^{-1}_{\varepsilon}(U_{ij})$,
where $F^{-1}_{\varepsilon}(u):=\inf\{x:F_{\varepsilon}(x)\geq u\}$
is the generalized inverse of $F_{\varepsilon}$. Using the fact that
$F^{-1}_{\varepsilon}(u)\leq c\iff u\leq F_{\varepsilon}(c)$, we
have $\ind\{\varepsilon_{ij}\leq c\}=\ind\{U_{ij}\leq F_{\varepsilon}(c)\}$
and hence 
\[
\hat{F}_{n}(c;z_{D})=G_{z_{D}}\big(F_{\varepsilon}(c)\big),
\]
where $G_{z_{D}}$ is defined as the empirical distribution function
of $\hat{m}(z_{D})$ i.i.d.\ uniform random variables at cell $z_{D}$.
The independence follows from the fact that each cell contains different
dyads across which $\varepsilon_{ij}$ is assumed to be independent
(Assumption \ref{ass:iid_errors}). Hence, $G_{z_{D}}$ has exactly
the same conditional distribution as the simulated $G^{(b)}_{z_{D}}$.
Substituting $\hat{F}_{n}(c;z_{D})=G_{z_{D}}\big(F_{\varepsilon}(c)\big)$
into \eqref{eq:Rn_range} and writing ${\cal R}:=F_{\varepsilon}(\R)\subseteq[0,1]$
for the range of $F_{\varepsilon}$, we have 
\[
\begin{aligned}R_{n} & =\sup_{u\in{\cal R}}\Big[\max_{z_{D}\in\hat{{\cal Z}}_{D}}G_{z_{D}}(u)-\min_{z_{D}\in\hat{{\cal Z}}_{D}}G_{z_{D}}(u)\Big]\\
 & \leq\sup_{u\in[0,1]}\Big[\max_{z_{D}\in\hat{{\cal Z}}_{D}}G_{z_{D}}(u)-\min_{z_{D}\in\hat{{\cal Z}}_{D}}G_{z_{D}}(u)\Big]=:R^{\ast}.
\end{aligned}
\]
The simulated $R^{(b)}$ then shares exactly the same distribution
as $R^{\ast}$. If $F_{\varepsilon}$ is continuous, then $R_{n}=R^{\ast}$.
When $F_{\varepsilon}$ has atoms, $R_{n}\neq R^{\ast}$ in general,
but $R_{n}\leq R^{\ast}$ is always true and thus $R^{\ast}$ remains
a valid device to control size as shown below. Since $(R^{\ast},R^{(1)},\dots,R^{(B)})$
is conditionally i.i.d.\ given $Z$, it is exchangeable, and thus
the ranking of $R^{\ast}$ in $(R^{\ast},R^{(1)},\dots,R^{(B)})$
is uniform over $\left\{ 1,\ldots,B+1\right\} $ up to ties. Therefore,
the finite-$B$-adjusted order statistic satisfies 
\[
\Prob\big(R^{\ast}>\hat{q}^{\,R}_{n,1-\alpha}(Z)\,\big|\,Z\big)\ \leq\ \frac{B+1-\lceil(1-\alpha)(B+1)\rceil}{B+1}\ \leq\ \alpha,
\]
where ties can only reduce the rejection probability. Since $\widehat{Q}_{n}(\theta_{0})\leq R_{n}$
by Lemma~\ref{lem:dom_semi} and $R_{n}\leq R^{\ast}$ as shown above,
we conclude that 
\[
\Prob\big(\widehat{Q}_{n}(\theta_{0})>\hat{q}^{\,R}_{n,1-\alpha}(Z)\,\big|\,Z\big)\leq\Prob\big(R^{\ast}>\hat{q}^{\,R}_{n,1-\alpha}(Z)\,\big|\,Z\big)\leq\alpha.
\]
 
\end{proof}

\begin{proof}[\textbf{Proof of Lemma~\ref{lem:dom_par}}]
Condition on $Z$. Fix a constant $c\in\R$ as before. Write $F_{\varepsilon}(c)=F_{\varepsilon}(c;\theta_{\varepsilon,0})$
in short for the true CDF. By the sandwich inequality \eqref{eq:empirical_sandwich}
and the definition of $R^{\mathrm{par}}_{n}\left(\theta_{\varepsilon,0}\right)$,
for every $z_{D}\in\hat{{\cal Z}}_{D}$, 
\begin{align*}
\hat{p}^{\,(n)}_{L}(z_{D},c;\theta_{0})-F_{\varepsilon}(c) & \ \leq\ \hat{F}_{n}(c;z_{D})-F_{\varepsilon}(c)\ \leq\ R^{\mathrm{par}}_{n}\left(\theta_{\varepsilon,0}\right),\\
F_{\varepsilon}(c)-\hat{p}^{\,(n)}_{U}(z_{D},c;\theta_{0}) & \ \leq\ F_{\varepsilon}(c)-\hat{F}_{n}(c;z_{D})\ \leq\ R^{\mathrm{par}}_{n}\left(\theta_{\varepsilon,0}\right).
\end{align*}
Taking maxima over $z_{D}$ and suprema over $c$ as before yields
the result. 
\end{proof}

\begin{proof}[\textbf{Proof of Theorem~\ref{thm:para_CS}}]
Condition on $Z$. As before, by Assumptions~\ref{ass:iid_errors}
and \ref{ass:exogeneity}, the real shock array $\{\varepsilon_{ij}\}_{ij}$
and the $B$ simulated arrays $\{\varepsilon^{(b)}_{ij}\}_{ij}$ are
i.i.d.\ draws from the same distribution given $Z$, so the realized
statistic $R^{\mathrm{par}}_{n}\left(\theta_{\varepsilon,0}\right)$
and the simulated statistics $R^{\mathrm{par},(1)}_{n}(\theta_{\varepsilon,0}),\ldots,R^{\mathrm{par},(B)}_{n}(\theta_{\varepsilon,0})$
are conditionally i.i.d.\ and hence exchangeable. Combined with Lemma
\ref{lem:dom_par}, the finite-$B$ rank argument proceeds in the
same way as in the proof of Theorem \ref{thm:semi_CS}. 
\end{proof}

\begin{proof}[\textbf{Proof of Theorem~\ref{prop:gendom}}]
At the true parameter $\overline{\theta}_{0}$,
\eqref{eq:empirical_sandwich} gives $V^{\pm}(z_{D},c;\overline{\theta}_{0})\leq R^{\pm}(z_{D},c;\theta_{\varepsilon,0})$
for every coordinate and every realization, so monotonicity of $T_{Z,\overline{\theta}_{0}}$
in each argument yields $T_{Z,\overline{\theta}_{0}}\big(V(\overline{\theta}_{0})\big)\leq T_{Z,\overline{\theta}_{0}}\big(R(\theta_{\varepsilon,0})\big)$.
Conditional on $Z$, the realized $T_{Z,\overline{\theta}_{0}}(R)$
and its simulated counterparts $T_{Z,\overline{\theta}_{0}}(R^{(1)}),\dots,T_{Z,\overline{\theta}_{0}}(R^{(B)})$
are again i.i.d., each distributed as $T_{Z,\overline{\theta}_{0}}\big(R(\theta_{\varepsilon,0})\big)$.
The rest of the proof is exactly the same as that of Theorem \ref{thm:semi_CS}. 
\end{proof}

\bibliographystyle{ecta}
\bibliography{SNF_NoFE}

\begin{thebibliography}{52}
\newcommand{\enquote}[1]{``#1''}
\expandafter\ifx\csname natexlab\endcsname\relax\def\natexlab#1{#1}\fi

\bibitem[\protect\citeauthoryear{Athey, Eckles, and Imbens}{Athey
  et~al.}{2018}]{athey2018exact}
\textsc{Athey, S., D.~Eckles, and G.~W. Imbens} (2018): \enquote{Exact p-values
  for network interference,} \emph{Journal of the American Statistical
  Association}, 113, 230--240.

\bibitem[\protect\citeauthoryear{Auerbach}{Auerbach}{2022{\natexlab{a}}}]{auerbach2022identification}
\textsc{Auerbach, E.} (2022{\natexlab{a}}): \enquote{Identification and
  estimation of a partially linear regression model using network data,}
  \emph{Econometrica}, 90, 347--365.

\bibitem[\protect\citeauthoryear{Auerbach}{Auerbach}{2022{\natexlab{b}}}]{auerbach2022testing}
---\hspace{-.1pt}---\hspace{-.1pt}--- (2022{\natexlab{b}}): \enquote{Testing
  for differences in stochastic network structure,} \emph{Econometrica}, 90,
  1205--1223.

\bibitem[\protect\citeauthoryear{Badev}{Badev}{2021}]{badev2021nash}
\textsc{Badev, A.} (2021): \enquote{Nash Equilibria on (Un){S}table Networks,}
  \emph{Econometrica}, 89, 1179--1206.

\bibitem[\protect\citeauthoryear{Barnard}{Barnard}{1963}]{barnard1963discussion}
\textsc{Barnard, G.~A.} (1963): \enquote{Discussion of Professor {B}artlett's
  paper,} \emph{Journal of the Royal Statistical Society: Series B
  (Methodological)}, 25, 294.

\bibitem[\protect\citeauthoryear{Beresteanu, Molchanov, and
  Molinari}{Beresteanu et~al.}{2011}]{beresteanu2011sharp}
\textsc{Beresteanu, A., I.~Molchanov, and F.~Molinari} (2011): \enquote{Sharp
  Identification Regions in Models With Convex Moment Predictions,}
  \emph{Econometrica}, 79, 1785--1821.

\bibitem[\protect\citeauthoryear{Berk and Jones}{Berk and
  Jones}{1979}]{berk1979goodness}
\textsc{Berk, R.~H. and D.~H. Jones} (1979): \enquote{Goodness-of-Fit Test
  Statistics That Dominate the {K}olmogorov Statistics,} \emph{Zeitschrift
  f{\"u}r Wahrscheinlichkeitstheorie und Verwandte Gebiete}, 47, 47--59.

\bibitem[\protect\citeauthoryear{Canay, Romano, and Shaikh}{Canay
  et~al.}{2017}]{canay2017randomization}
\textsc{Canay, I.~A., J.~P. Romano, and A.~M. Shaikh} (2017):
  \enquote{Randomization tests under an approximate symmetry assumption,}
  \emph{Econometrica}, 85, 1013--1030.

\bibitem[\protect\citeauthoryear{Chesher}{Chesher}{2010}]{chesher2010instrumental}
\textsc{Chesher, A.} (2010): \enquote{Instrumental Variable Models for Discrete
  Outcomes,} \emph{Econometrica}, 78, 575--601.

\bibitem[\protect\citeauthoryear{Chesher and Rosen}{Chesher and
  Rosen}{2017}]{chesher2017generalized}
\textsc{Chesher, A. and A.~M. Rosen} (2017): \enquote{Generalized Instrumental
  Variable Models,} \emph{Econometrica}, 85, 959--989.

\bibitem[\protect\citeauthoryear{Comola and Dekel}{Comola and
  Dekel}{2026}]{comola2026estimating}
\textsc{Comola, M. and A.~Dekel} (2026): \enquote{Estimating Network
  Externalities in Undirected Link Formation Games,} \emph{Journal of Applied
  Econometrics}, forthcoming.

\bibitem[\protect\citeauthoryear{de~Paula}{de~Paula}{2020}]{de2020econometric}
\textsc{de~Paula, {\'A}.} (2020): \enquote{Econometric Models of Network
  Formation,} \emph{Annual Review of Economics}, 12, 775--799.

\bibitem[\protect\citeauthoryear{de~Paula, Richards-Shubik, and Tamer}{de~Paula
  et~al.}{2018}]{de2018identifying}
\textsc{de~Paula, {\'A}., S.~Richards-Shubik, and E.~Tamer} (2018):
  \enquote{Identifying Preferences in Networks With Bounded Degree,}
  \emph{Econometrica}, 86, 263--288.

\bibitem[\protect\citeauthoryear{de~Paula and Tang}{de~Paula and
  Tang}{2012}]{de2012inference}
\textsc{de~Paula, {\'A}. and X.~Tang} (2012): \enquote{Inference of Signs of
  Interaction Effects in Simultaneous Games With Incomplete Information,}
  \emph{Econometrica}, 80, 143--172.

\bibitem[\protect\citeauthoryear{Dufour}{Dufour}{2006}]{dufour2006monte}
\textsc{Dufour, J.-M.} (2006): \enquote{Monte {C}arlo tests with nuisance
  parameters: A general approach to finite-sample inference and nonstandard
  asymptotics,} \emph{Journal of Econometrics}, 133, 443--477.

\bibitem[\protect\citeauthoryear{Dwass}{Dwass}{1957}]{dwass1957modified}
\textsc{Dwass, M.} (1957): \enquote{Modified randomization tests for
  nonparametric hypotheses,} \emph{Annals of Mathematical Statistics}, 28,
  181--187.

\bibitem[\protect\citeauthoryear{Dzemski}{Dzemski}{2019}]{dzemski2019empirical}
\textsc{Dzemski, A.} (2019): \enquote{An Empirical Model of Dyadic Link
  Formation in a Network With Unobserved Heterogeneity,} \emph{Review of
  Economics and Statistics}, 101, 763--776.

\bibitem[\protect\citeauthoryear{Eagleson and Weber}{Eagleson and
  Weber}{1978}]{eagleson1978limit}
\textsc{Eagleson, G.~K. and N.~C. Weber} (1978): \enquote{Limit Theorems for
  Weakly Exchangeable Arrays,} \emph{Mathematical Proceedings of the Cambridge
  Philosophical Society}, 84, 123--130.

\bibitem[\protect\citeauthoryear{Galichon and Henry}{Galichon and
  Henry}{2011}]{galichon2011set}
\textsc{Galichon, A. and M.~Henry} (2011): \enquote{Set Identification in
  Models With Multiple Equilibria,} \emph{The Review of Economic Studies}, 78,
  1264--1298.

\bibitem[\protect\citeauthoryear{Gao}{Gao}{2020}]{gao2020}
\textsc{Gao, W.~Y.} (2020): \enquote{Nonparametric Identification in Index
  Models of Link Formation,} \emph{Journal of Econometrics}, 215, 399--413.

\bibitem[\protect\citeauthoryear{Gao, Li, and Xu}{Gao
  et~al.}{2023}]{gao2023logical}
\textsc{Gao, W.~Y., M.~Li, and S.~Xu} (2023): \enquote{Logical Differencing in
  Dyadic Network Formation Models With Nontransferable Utilities,}
  \emph{Journal of Econometrics}, 235, 302--324.

\bibitem[\protect\citeauthoryear{Gao, Li, and Xu}{Gao
  et~al.}{2026}]{gao2026tractable}
\textsc{Gao, W.~Y., M.~Li, and Z.~Xu} (2026): \enquote{Tractable Identification
  of Strategic Network Formation Models With Unobserved Heterogeneity,}
  {arXiv}:2603.08634.

\bibitem[\protect\citeauthoryear{Gao and Wang}{Gao and
  Wang}{2026}]{gao2026identification}
\textsc{Gao, W.~Y. and R.~Wang} (2026): \enquote{Identification in Nonlinear
  Dynamic Panel Models under Partial Stationarity,} \emph{Journal of
  Econometrics}, 253, 106185.

\bibitem[\protect\citeauthoryear{Graham}{Graham}{2017}]{graham2017econometric}
\textsc{Graham, B.~S.} (2017): \enquote{An Econometric Model of Network
  Formation With Degree Heterogeneity,} \emph{Econometrica}, 85, 1033--1063.

\bibitem[\protect\citeauthoryear{Graham and Pelican}{Graham and
  Pelican}{2020}]{grahampelican2020testing}
\textsc{Graham, B.~S. and A.~Pelican} (2020): \enquote{Testing for
  Externalities in Network Formation Using Simulation,} in \emph{The
  Econometric Analysis of Network Data}, ed. by B.~S. Graham and
  {\'A}.~de~Paula, Cambridge, MA: Academic Press, 63--82.

\bibitem[\protect\citeauthoryear{Gualdani}{Gualdani}{2021}]{gualdani2021identification}
\textsc{Gualdani, C.} (2021): \enquote{An Econometric Model of Network
  Formation With an Application to Board Interlocks Between Firms,}
  \emph{Journal of Econometrics}, 224, 345--370.

\bibitem[\protect\citeauthoryear{Jackson and Wolinsky}{Jackson and
  Wolinsky}{1996}]{jackson1996strategic}
\textsc{Jackson, M.~O. and A.~Wolinsky} (1996): \enquote{A Strategic Model of
  Social and Economic Networks,} \emph{Journal of Economic Theory}, 71, 44--74.

\bibitem[\protect\citeauthoryear{Jochmans}{Jochmans}{2018}]{jochmans2017semiparametric}
\textsc{Jochmans, K.} (2018): \enquote{Semiparametric Analysis of Network
  Formation,} \emph{Journal of Business \& Economic Statistics}, 36, 705--713.

\bibitem[\protect\citeauthoryear{Kojevnikov, Marmer, and Song}{Kojevnikov
  et~al.}{2021}]{kojevnikov2021limit}
\textsc{Kojevnikov, D., V.~Marmer, and K.~Song} (2021): \enquote{Limit Theorems
  for Network Dependent Random Variables,} \emph{Journal of Econometrics}, 222,
  882--908.

\bibitem[\protect\citeauthoryear{Leung}{Leung}{2015}]{leung2015twostep}
\textsc{Leung, M.~P.} (2015): \enquote{Two-Step Estimation of Network-Formation
  Models With Incomplete Information,} \emph{Journal of Econometrics}, 188,
  182--195.

\bibitem[\protect\citeauthoryear{Leung}{Leung}{2019}]{leung2019treatment}
---\hspace{-.1pt}---\hspace{-.1pt}--- (2019): \enquote{A Weak Law for Moments
  of Pairwise-Stable Networks,} \emph{Journal of Econometrics}, 210, 310--326.

\bibitem[\protect\citeauthoryear{Leung}{Leung}{2022}]{leung2022causal}
---\hspace{-.1pt}---\hspace{-.1pt}--- (2022): \enquote{Causal Inference Under
  Approximate Neighborhood Interference,} \emph{Econometrica}, 90, 267--293.

\bibitem[\protect\citeauthoryear{Leung}{Leung}{2023}]{leung2023network}
---\hspace{-.1pt}---\hspace{-.1pt}--- (2023): \enquote{Network Cluster-Robust
  Inference,} \emph{Econometrica}, 91, 641--667.

\bibitem[\protect\citeauthoryear{Leung and Moon}{Leung and
  Moon}{2026}]{leung2026normal}
\textsc{Leung, M.~P. and H.~R. Moon} (2026): \enquote{Normal Approximation in
  Large Network Models,} \emph{The Review of Economic Studies}, forthcoming.

\bibitem[\protect\citeauthoryear{Li and Henry}{Li and
  Henry}{2026}]{li2026finite}
\textsc{Li, L. and M.~Henry} (2026): \enquote{Finite Sample Inference in
  Incomplete Models,} {arXiv}:2204.00473.

\bibitem[\protect\citeauthoryear{Li, Shi, and Zheng}{Li
  et~al.}{2026}]{li2026bagging}
\textsc{Li, M., Z.~Shi, and Y.~Zheng} (2026): \enquote{Bagging the Network,}
  {arXiv}:2410.23852.

\bibitem[\protect\citeauthoryear{Liu, de~Paula, and Tamer}{Liu
  et~al.}{2026}]{liu2026prediction}
\textsc{Liu, W., {\'A}.~de~Paula, and E.~Tamer} (2026): \enquote{Prediction
  Sets and Conformal Inference With Interval Outcomes,} {arXiv}:2501.10117.

\bibitem[\protect\citeauthoryear{Manski}{Manski}{1975}]{manski1975maximum}
\textsc{Manski, C.~F.} (1975): \enquote{Maximum Score Estimation of the
  Stochastic Utility Model of Choice,} \emph{Journal of Econometrics}, 3,
  205--228.

\bibitem[\protect\citeauthoryear{Manski}{Manski}{1985}]{manski1985semiparametric}
---\hspace{-.1pt}---\hspace{-.1pt}--- (1985): \enquote{Semiparametric Analysis
  of Discrete Response: Asymptotic Properties of the Maximum Score Estimator,}
  \emph{Journal of Econometrics}, 27, 313--333.

\bibitem[\protect\citeauthoryear{Mastrandrea, Fournet, and Barrat}{Mastrandrea
  et~al.}{2015}]{mastrandrea2015contact}
\textsc{Mastrandrea, R., J.~Fournet, and A.~Barrat} (2015): \enquote{Contact
  Patterns in a High School: A Comparison Between Data Collected Using Wearable
  Sensors, Contact Diaries and Friendship Surveys,} \emph{PLOS ONE}, 10,
  e0136497.

\bibitem[\protect\citeauthoryear{Mele}{Mele}{2017}]{mele2017dense}
\textsc{Mele, A.} (2017): \enquote{A Structural Model of Dense Network
  Formation,} \emph{Econometrica}, 85, 825--850.

\bibitem[\protect\citeauthoryear{Menzel}{Menzel}{2021}]{menzel2021clt}
\textsc{Menzel, K.} (2021): \enquote{Central Limit Theory for Models of
  Strategic Network Formation,} Working paper, New York University.

\bibitem[\protect\citeauthoryear{Menzel}{Menzel}{2026}]{menzel2026strategic}
---\hspace{-.1pt}---\hspace{-.1pt}--- (2026): \enquote{Strategic Network
  Formation With Many Agents,} \emph{Journal of Econometrics}, 253, 106174.

\bibitem[\protect\citeauthoryear{Miyauchi}{Miyauchi}{2016}]{miyauchi2016structural}
\textsc{Miyauchi, Y.} (2016): \enquote{Structural Estimation of Pairwise Stable
  Networks With Nonnegative Externality,} \emph{Journal of Econometrics}, 195,
  224--235.

\bibitem[\protect\citeauthoryear{Pelican and Graham}{Pelican and
  Graham}{2022}]{pelican2022optimal}
\textsc{Pelican, A. and B.~S. Graham} (2022): \enquote{An Optimal Test for
  Strategic Interaction in Social and Economic Network Formation Between
  Heterogeneous Agents,} NBER Working Paper 27793; arXiv:2009.00212.

\bibitem[\protect\citeauthoryear{Ridder and Sheng}{Ridder and
  Sheng}{2025}]{ridder2025twostep}
\textsc{Ridder, G. and S.~Sheng} (2025): \enquote{Two-Step Estimation of a
  Strategic Network Formation Model With Clustering,} {arXiv}:2001.03838.

\bibitem[\protect\citeauthoryear{Rosen and Ura}{Rosen and
  Ura}{2025}]{rosen2025finite}
\textsc{Rosen, A.~M. and T.~Ura} (2025): \enquote{Finite Sample Inference for
  the Maximum Score Estimand,} \emph{The Review of Economic Studies}, 92,
  4117--4151.

\bibitem[\protect\citeauthoryear{Rozemberczki, Allen, and Sarkar}{Rozemberczki
  et~al.}{2021}]{rozemberczki2021multiscale}
\textsc{Rozemberczki, B., C.~Allen, and R.~Sarkar} (2021): \enquote{Multi-Scale
  Attributed Node Embedding,} \emph{Journal of Complex Networks}, 9, cnab014.

\bibitem[\protect\citeauthoryear{Sheng}{Sheng}{2020}]{sheng2020structural}
\textsc{Sheng, S.} (2020): \enquote{A Structural Econometric Analysis of
  Network Formation Games Through Subnetworks,} \emph{Econometrica}, 88,
  1829--1858.

\bibitem[\protect\citeauthoryear{Tamer}{Tamer}{2003}]{tamer2003incomplete}
\textsc{Tamer, E.} (2003): \enquote{Incomplete Simultaneous Discrete Response
  Model With Multiple Equilibria,} \emph{The Review of Economic Studies}, 70,
  147--165.

\bibitem[\protect\citeauthoryear{Wu}{Wu}{2024}]{wu2024twostep}
\textsc{Wu, S.} (2024): \enquote{Two-Step Estimation of Network Formation
  Models With Unobserved Heterogeneities and Strategic Interactions,}
  {arXiv}:2404.12581.

\bibitem[\protect\citeauthoryear{Zeleneev}{Zeleneev}{2026}]{zeleneev2026identification}
\textsc{Zeleneev, A.} (2026): \enquote{Identification and Estimation of Network
  Models With Nonparametric Unobserved Heterogeneity,} {arXiv}:2602.06885.

\end{thebibliography}

\end{document}